\documentclass[letterpaper, twocolumn, accepted=2024-06-24]{quantumarticle}

\pdfoutput=1

\usepackage[numbers,sort&compress]{natbib}
\usepackage{amsmath, amssymb, amsfonts, amsthm, bm, bbm, graphicx, listings, mathtools, enumerate, stmaryrd, tikz-cd, tikz, enumitem}
\usepackage{booktabs,stackrel,eczoo}
\usepackage[colorlinks=true, linkcolor=blue, citecolor=blue, urlcolor=blue]{hyperref}
\usetikzlibrary{calc, positioning}

\DeclareMathOperator{\spn}{Span}

\newlength{\nd}
\usepackage{suffix}

\newcommand{\set}[1]{\left\{ #1 \right\}}
\newcommand{\abs}[1]{\left| #1 \right|}
\newcommand{\parenth}[1]{\left( #1 \right)}
\newcommand{\generate}[1]{\left\langle #1 \right\rangle}

\newcommand{\om}[0]{\omega}

\newcommand{\De}[0]{\Delta}

\newcommand{\ket}[1]{\left| #1 \right\rangle}

\newcommand{\dsC}[0]{\mathbb{C}}

\newcommand{\dsF}[0]{\mathbb{F}}
\newcommand{\dsG}[0]{\mathbb{G}}
\newcommand{\dsH}[0]{\mathbb{H}}

\newcommand{\dsL}[0]{\mathbb{L}}

\newcommand{\dsN}[0]{\mathbb{N}}

\newcommand{\dsZ}[0]{\mathbb{Z}}

\newcommand{\scC}[0]{\mathcal{C}}

\newcommand{\scG}[0]{\mathcal{G}}

\newcommand{\scP}[0]{\mathcal{P}}

\newcommand{\scS}[0]{\mathcal{S}}
\newcommand{\scT}[0]{\mathcal{T}}
\newcommand{\scU}[0]{\mathcal{U}}
\newcommand{\scV}[0]{\mathcal{V}}
\newcommand{\scW}[0]{\mathcal{W}}

\newcommand{\goursatCorrespondence}[7]{
\begin{array}{ccccc}
 & & #1 \times #1 & &  \\
 #1\times #1 & & | & & \\
 | & \longleftrightarrow & #3 \times #4 &,& #7 \\
 #2 & & | & & \\
 & & #5 \times #6 & & 
\end{array}
}

\newtheorem{proposition}{Proposition}
\newtheorem{lemma}{Lemma}
\newtheorem{theorem}{Theorem}
\newtheorem{corollary}{Corollary}

\theoremstyle{definition}
\newtheorem{definition}{Definition}
\newtheorem{example}{Example}

\DeclareMathOperator{\wt}{wt}
\DeclareMathOperator{\swt}{swt}
\DeclareMathOperator{\End}{End}
\DeclareMathOperator{\Syn}{Syn}
\DeclareMathOperator{\Par}{Par}
\DeclareMathOperator{\Alg}{Alg}

\newcommand{\vva}[1]{ { \color{red} (VVA: {#1}) }}
\WithSuffix\newcommand\vva*[1]{{\color{red} #1}}

\begin{document}

\title{Subsystem CSS codes, a tighter stabilizer-to-CSS mapping, and Goursat's Lemma}

\author{Michael Liaofan Liu}
\email{mliu24@amherst.edu}
\affiliation{Joint Center for Quantum Information and Computer Science, NIST, and University of Maryland, College Park, Maryland 20740, USA}
\affiliation{Department of Mathematics, Amherst College, Amherst, Massachusetts 01002, USA}

\author{Nathanan Tantivasadakarn}
\affiliation{Walter Burke Institute for Theoretical Physics and Department of Physics, California Institute of Technology, Pasadena, California 91125, USA}

\author{Victor V. Albert}
\affiliation{Joint Center for Quantum Information and Computer Science, NIST, and University of Maryland, College Park, Maryland 20740, USA}

\begin{abstract}
% We study subsystem CSS codes, i.e., those subsystem stabilizer codes whose gauge group can be generated by $X$-type and $Z$-type Pauli operators. First, we show that every subsystem stabilizer code can be ``doubled" to yield a subsystem CSS code with twice the number of physical, logical, and gauge qudits and up to twice the code distance. 
% Next, we present a two-step recovery procedure for subsystem CSS codes: in the first step, we determine the syndrome of a Pauli error using gauge generator measurements, and in the second step, we correct the Pauli error up to gauge terms using two classical linear codes associated with the parent subsystem CSS code. 
% Finally, we show that every subsystem stabilizer code can be constructed from two nested subsystem CSS codes subject to certain constraints. Applying this theory, we investigate the structure of a given subsystem stabilizer code in terms of its underlying subsystem CSS codes, and we introduce two generalizations of subsystem CSS codes.

The CSS code construction is a powerful framework used to express features of a quantum code in terms of a pair of underlying classical codes.
Its subsystem extension allows for similar expressions, but the general case has not been fully explored.
Extending previous work of Aly, Klappenecker, and Sarvepalli~[\href{https://arxiv.org/abs/quant-ph/0610153}{quant-ph/0610153}], we determine subsystem CSS code parameters, express codewords, and develop a Steane-type decoder using only data from the two underlying classical codes.
Generalizing a result of Kovalev and Pryadko~[\href{https://doi.org/10.1103/PhysRevA.88.012311}{Phys.~Rev.~A \textbf{88} 012311 (2013)}], we show that any subsystem stabilizer code can be ``doubled" to yield a subsystem CSS code with twice the number of physical, logical, and gauge qudits and up to twice the code distance.
%\NT{Should we add here that "for qubit subspace codes"} \vva{personally I wouldn't bother since it seems not hard to extend it, e.g., Nico did it in his bachelor thesis}
This mapping preserves locality and is tighter than the Majorana-based mapping of Bravyi, Terhal, and Leemhuis~[\href{https://doi.org/10.1088/1367-2630/12/8/083039}{New J.~Phys.~\textbf{12} 083039 (2010)}].
% Applying this theory, we investigate the structure of a given subsystem stabilizer code in terms of its underlying subsystem CSS codes, and we introduce two generalizations of subsystem CSS codes.
Using Goursat's Lemma, we show that every subsystem stabilizer code can be constructed from two nested subsystem CSS codes satisfying certain constraints, and we characterize subsystem stabilizer codes based on the nested codes' properties.
% , and we develop algorithms to determine whether a subsystem stabilizer code has a CSS-type gauge or stabilizer group.

\end{abstract}

% Express properties of general CSS subsystem codes in terms of the two underlying classical linear codes used for the construction.
% Such properties include the code parameters, codewords,
% Map \([[]]\) subsystem codes to \(\) stabilizer codes in a mapping that improves upon the Majorana-based mapping of [Bravyi ...].
% 

\maketitle

% \tableofcontents

\section{Introduction and summary of results} \label{sec:Introduction}

An extensive theory of quantum error-correcting codes has been developed to protect quantum computers from noise \cite{KL97, AB08, S95}
%KLPL06, BKK07
\eczoo{qecc}. Qualitatively, a quantum error-correcting code describes how to ``hide" quantum information within a protected subsystem of a quantum system, such that noise within the quantum system can be detected and removed from the protected subsystem. The class of quantum error-correcting codes called \emph{stabilizer codes} \cite{G97, CRSS97} \eczoo{stabilizer} remains the most promising route to a working and robust quantum computer, in part due to extra structure that is useful for promptly detecting and correcting errors. 

We are interested in the special class of \emph{subsystem stabilizer codes} \cite{P05, KS08} \eczoo{galois_subsystem_stabilizer}, which correspond to the nontrivial normal subgroups of a Pauli group.
While subsystem stabilizer codes can be derived from subspace stabilizer codes by using only a subset of the original logical qudits to store quantum information, such codes can exhibit advantageous and/or inherently different properties from subspace stabilizer codes such as lower weights of check operators~\cite{BDPS13}, new fault-tolerant protocols~\cite{Anderson14}, the association with more general types of anyon theories in two dimensions~\cite{Ellison23}, as well as single-shot error correction in three dimensions~\cite{Bombin15,KV22}.
The extra data used to define subsystem stabilizer codes yields a richer code structure, warranting its own investigation.

A large part of our work studies \emph{subsystem CSS codes} \cite{AKS06, AK08} \eczoo{galois_subsystem_css}, where the corresponding normal subgroup admits a generating set consisting of \(X\)-type and \(Z\)-type Pauli operators.
Subsystem CSS codes generalize subspace CSS codes \cite{CS96, S96, S96B} \eczoo{qubit_css}, of which there are many useful examples \cite{K03, BM06}, recently culminating in the first asymptotically good quantum low-density parity-check codes \cite{FH21, P23, DHLV23}. Subspace CSS codes have rich connections to classical coding theory \cite{CRSS98, H18, S99} and homology theory \cite{BH14, BM07, B18}, and they admit a recovery procedure that utilizes classical linear codes to independently correct $X$-type and $Z$-type Pauli errors \cite{CS96, NC10}. Thus, it is interesting to investigate the generalization of subspace CSS codes to the subsystem setting. Indeed, many subsystem stabilizer codes are CSS \cite{B06, BC06, CZYHC20}, so a general theory of subsystem CSS codes would apply to a broad class of examples developed to date.

To study subsystem CSS codes, we adopt a linear-algebraic perspective that streamlines previous proofs and distills the essential mathematical ideas \cite{Haah13}. 
This places subsystem CSS codes on the same footing as subspace CSS codes, allowing us to generalize many important features of subspace CSS codes to the subsystem setting. 

Our main tool is the symplectic representation, which allows us to view a subsystem stabilizer code as a subspace of a vector space by considering Pauli groups modulo phase factors. In this representation, the structure and parameters of a subsystem stabilizer code are encoded in a ``tower'' of subspaces, which can then be manipulated using tools from linear algebra \cite{TA86,MG22,AC09} to yield several general results.

Our results include a mapping from subsystem stabilizer codes to subsystem CSS codes, a fleshing out of the connection between subsystem CSS codes and classical coding theory, a general recovery procedure for subsystem CSS codes, and a characterization of subsystem stabilizer codes based on Goursat's Lemma from group theory.
% In particular, our perspective should be helpful to mathematicians seeking to apply tools from abstract algebra to produce various mappings of subsystem stabilizer codes, since such codes are represented as towers of algebraic objects in the linear algebraic formalism.

% This work is organized as follows. 
In Section \ref{sec:Prelim}, we review the theory of Pauli groups and subsystem stabilizer codes, and we explain how to translate these objects to the linear algebraic framework.

In Section \ref{sec:subsystemCSS}, we review the subspace CSS construction and its generalization to the subsystem case.
While the subspace CSS construction utilizes two classical linear codes \cite{HKS21} \eczoo{q-ary_linear} under certain constraints, the subsystem CSS construction utilizes two classical linear codes with no constraints.
We express the logical and gauge codewords in terms of the two underlying classical linear codes.

In Section \ref{sec:Mapping}, we show that the performance of subsystem CSS codes is comparable to that of subsystem stabilizer codes. More precisely, we show that every modular-qudit subsystem stabilizer code \eczoo{qudit_subsystem_stabilizer} can be ``doubled" to yield a subsystem CSS code with twice the number of physical, logical, and gauge qudits and up to twice the code distance. This generalizes the mapping for subspace stabilizer codes obtained in \cite[Theorem 1]{KP13}.

In Section \ref{sec:CSSRecovery}, we present an error-correction procedure for subsystem CSS codes that generalizes the well-known Steane procedure for subspace CSS codes (see, e.g., \cite{CS96,NC10,G10}). 
% While procedures for many examples of subsystem CSS codes 
While specific instances of this procedure have been presented previously \cite{BC15,GS18,KV22,BNB16,BDPS13,HB21,ABKM12}, we are unaware of a presentation of the general case that uses only information from the underlying classical linear codes. 
Our error-correction procedure for subsystem CSS codes decodes up to arbitrary actions on the gauge qudits, and the corresponding classical step in our error-correction procedure decodes up to prescribed ``redundant" subcodes obtained from the underlying classical codes.
% The first step in our recovery procedure is purely quantum: we determine the syndrome of a Pauli error using gauge generator measurements. The second step in our recovery procedure is purely classical: we use two classical linear codes to independently recover the $X$-type and $Z$-type components of the Pauli error up to gauge terms. 
We explore an alternative decoder that generalizes the recovery procedure for the Bacon--Shor code \cite{B06} in Appendix \ref{app:Par}. 

% This alternative second step uses a related pair of classical linear codes to recover the Pauli error up to gauge terms.

In Section \ref{sec:SubsystemStabilizerCode}, we show that every subsystem stabilizer code can be constructed from two subsystem CSS codes satisfying certain constraints, much as every subsystem CSS code can be constructed from two classical codes. 
We investigate the structure of a given subsystem stabilizer code in terms of its two underlying subsystem CSS codes, and we define two code families that can be viewed as generalizations of subsystem CSS codes from this perspective (see Fig.~\ref{fig:Figure2}).

Finally, we provide concluding remarks in Section \ref{sec:Conclusion}. Most proofs have been deferred to Appendix \ref{app:Proofs} for readability.

\section{Preliminaries} \label{sec:Prelim}

In Section \ref{sec:PrelimPauli}, we review the construction of Pauli groups. In Section \ref{sec:PrelimPauliVS}, we explain how to interpret Pauli groups as vector spaces, and we review the related theory of bilinear forms on vector spaces. In Section \ref{sec:PrelimSSC}, we review the construction of subsystem stabilizer codes. Finally, in Section \ref{sec:PrelimSSCVS}, we explain how to interpret subsystem stabilizer codes as subspaces of vector spaces, and we define the properties of subsystem stabilizer codes in the vector space setting.

\subsection{Pauli groups}\label{sec:PrelimPauli}

Consider a quantum system with $n$ prime qudits, each of dimension $p$. The Hilbert space of this system is given by the quantization of \(G\coloneqq \dsF_p^n\), i.e.,
\begin{equation}
\dsH\coloneqq \left\langle \{\ket{g} \mid g\in G\}\right\rangle \cong {\parenth{\dsC^p}}^{\otimes n}.
\end{equation}
Acting on this space is a special class of linear operators called \emph{Pauli operators}, denoted by \(X^a,Z^b\) for any \(a,b\in G\). For any basis vector \(\ket{g} \in \dsH\), the Pauli operators act as
\begin{align}
    X^a\ket{g} &\coloneqq \ket{a+g} \text{ and} \\
    Z^b\ket{g} &\coloneqq e^{\frac{2\pi i}{p}b \cdot g}\ket{g},
\end{align}
where $\cdot$ is the usual dot product on $G$. One verifies the relations
\begin{align}
 \label{eq:PauliPropertiesIso}X^aZ^bX^cZ^d &=e^{\frac{2\pi i}{p}b \cdot c}X^{a+c}Z^{b+d} \text{ and}
 \\ X^aZ^bX^cZ^d&=e^{\frac{2\pi
 i}{p}(b \cdot c - a \cdot d)}X^cZ^dX^aZ^b
 \label{eq:PauliPropertiesCommutation}
\end{align}
among the Pauli operators. We define the \emph{weight} of a Pauli operator to be the number of qudits on which it acts nontrivially. The Pauli operators generate the \emph{Pauli group}, which is given explicitly by
\begin{equation}\label{eq:PauliGroupDefinition}
\begin{aligned}
 \scP &\coloneqq \left\langle \left\{X^a,Z^b \, \middle| \, a,b\in G\right\}\right\rangle \\ &= \left\{e^{\frac{2\pi i}{p}\kappa} X^aZ^b \, \middle| \, \kappa \in \dsF_p,a,b\in G\right\}.
\end{aligned}
\end{equation}
We remark that in other works, the phase factors included in the Pauli group may differ depending on whether $p$ is odd or even \cite{M19, Ellison23, DKV24}. This will not matter much for our work, since we will henceforth consider Pauli groups modulo phase factors.

\subsection{Pauli groups as vector spaces}\label{sec:PrelimPauliVS}

In light of Eqs.~(\ref{eq:PauliPropertiesIso}) and (\ref{eq:PauliGroupDefinition}), one sees that the Pauli group modulo phases is essentially equivalent to the vector space $G \times G$. Formally, letting
\begin{equation}
\Phi \coloneqq \left\{e^{\frac{2\pi i}{p}\kappa} \, \middle| \, \kappa \in \dsF_p\right\}
\end{equation}
denote the collection of phases in the Pauli group, one verifies that the map
\begin{equation}
\begin{aligned}\label{eq:gtimesgiso}
    \scP / \Phi &\to G \times G \\
    \Phi X^aZ^b &\mapsto (a,b)
\end{aligned}
\end{equation}
is a vector space isomorphism. For most of this work, it will be convenient to view the Pauli group modulo phases as the vector space $G \times G$. We review the necessary theory below \cite{Haah13,TA86}.

A function \(\xi: G\times G \to \dsF_p\) is called \emph{bilinear} if
\begin{equation}
\begin{aligned}
 \xi (a+\lambda \, b,c) &=\xi (a,c)+\lambda \, \xi (b,c) \text{ and} \\
 \xi (a,b+\lambda \, c) &=\xi (a,b)+\lambda \, \xi (a,c),
\end{aligned}
\end{equation}
\emph{symmetric} if
 \begin{equation}\xi (a,b)=\xi (b,a),\end{equation}
\emph{antisymmetric} if
 \begin{equation}\xi (a,b)=-\xi (b,a),\end{equation}
and \emph{nondegenerate} if
 \begin{equation}\xi (a,G)=0\implies a=0.\end{equation}
We call \(\xi\) a \emph{form} if it is bilinear and nondegenerate. Henceforth, $\xi$ will denote a symmetric or antisymmetric form, $\theta$ will denote a symmetric form, and $\om$ will denote an antisymmetric form. Every subspace \(H\leq G\) has a \emph{complement with respect to} $\xi$
\begin{equation}H^{\xi} \coloneqq \left\{a\in G \, \middle| \, \xi (a,H)=0\right\}\leq G.\end{equation}
One verifies that
\begin{align}\dim H+\dim H^{\xi} &=\dim G \text{ and} \label{eq:xicomplementdim}\\
\left(H^\xi\right)^\xi&=H.\end{align}
Moreover, if $K \leq G$ is another subspace, then
\begin{align}(H+K)^{\xi} &=H^{\xi} \cap K^{\xi} \text{ and} \\(H\cap K)^{\xi} &=H^{\xi} +K^{\xi}.\end{align}
A symmetric form $\theta$ on $G$ passes naturally to a symmetric form
\begin{equation}((a,b),(c,d)) \mapsto \theta (a,c)+\theta (b,d)\end{equation}
on $G \times G$, which by abuse of notation we also call $\theta$ (note that for the Pauli group, $\theta$ is the dot product). In addition, the original symmetric form $\theta$ on $G$ induces the antisymmetric form 
\begin{equation}
\om((a,b),(c,d)) \coloneqq \theta (b,c)-\theta (a,d)
\end{equation}
on $G \times G$ (note that for the Pauli group, $\om$ encodes the commutation of the Pauli operators as in Eq.~(\ref{eq:PauliPropertiesCommutation})). Direct products behave well with respect to these induced forms. Indeed, if \(H_X\times H_Z\leq G\times G\), then
\begin{align}
\left(H_X\times H_Z\right)^{\theta} &= H_X^{\theta
}\times H_Z^{\theta} \text{ and} \\ \label{eq:DirectProductOmega}\left(H_X\times H_Z\right)^{\om} &= H_Z^{\theta
}\times H_X^{\theta}.\end{align}

For any $a = (a_1,\ldots,a_n),b = (b_1,\ldots,b_n) \in G$, we call
\begin{equation}\wt(a)\coloneqq \left| \left\{j \in \set{1,\ldots,n} \mid a_j \neq 0\right\}\right| \end{equation}
the \emph{weight} of $a$, and we call
\begin{equation}\swt(a,b)\coloneqq \left| \left\{j\in \set{1,\ldots,n} \mid \left(a_j,b_j\right)\neq (0,0)\right\}\right| \end{equation}
the \emph{symplectic weight} of \((a,b)\). For any nontrivial \(S\subseteq G\) and \(T\subseteq G\times G\), we write
\begin{align}
\min \wt(S)&\coloneqq \underset{a\in S\setminus \{0\}}{\min}\wt(a) \text{ and} \\ \min \swt(T)&\coloneqq \underset{(a,b)\in T\setminus \{0\}}{\min} \swt(a,b).
\end{align}
Note that the weight of a Pauli operator $X^a Z^b$ is equal to the symplectic weight of the corresponding vector $(a,b)$.

Finally, for any algebraic objects $G,H$, we write
\begin{equation}\begin{tikzpicture}
    [
    node distance = \nd and \nd,
    on grid,
    baseline = (current bounding box.center)
    ]

    \node (1) {\(G\)};
    \node (2) [below=of 1] {\(H\)};
    
    \draw
        (1) -- node [right] {$l$}
        (2);
\end{tikzpicture}\end{equation}
to indicate that \(H\leq G\) and, if $G,H$ are vector spaces, that \(\dim G/H=l\).

\subsection{Subsystem stabilizer codes} \label{sec:PrelimSSC}

A subsystem stabilizer code can be interpreted as a subspace stabilizer code with some of its logical qudits relegated to gauge qudits.
This necessitates the use of a ``gauge group'' to determine the gauge qudits.

In the usual formalism for subsystem stabilizer codes \cite{P05,DKV24}, one begins with an abelian ``stabilizer group" $\scS \leq \scP$ whose intersection with $\Phi$ is trivial. Then, one considers the centralizer of $\scS$, denoted $\scC(\scS)$, which is the subgroup of all Pauli operators that commute with everything in $\scS$. Finally, one selects a ``gauge group" $\scG \leq \scC(\scS)$ whose center is $\generate{\Phi, \scS}$. These relationships can be summarized in the tower of normal subgroups
\begin{equation}\label{eq:PauliTower}
\begin{tikzpicture}
    [
    node distance = \nd and \nd,
    on grid,
    baseline = (current bounding box.center)
    ]

    \node (1) {\(\scP\)};
    \node (2) [below=of 1] {\( \scC(\scS)\)};
    \node (3) [below=of 2] {\(\scG\)};
    \node (4) [below=of 3] {\( \generate{\Phi,\scS}\)};
    \node (5) [below=of 4] {\(\Phi\)};
    
    \draw
        (1) -- 
        (2) -- 
        (3) -- 
        (4) -- 
        (5);
\end{tikzpicture}
\end{equation}
which encodes the parameters of the subsystem stabilizer code. Indeed, the number of logical (resp. gauge) qudits in $\scG$ is given by half the size of any minimal generating set for $\scC(\scS) / \scG$ (resp. $\scG / {\generate{\Phi,\scS}}$). Moreover, the distance of $\scG$ is defined to be the minimum weight of a Pauli operator in $\scC(\scS) \setminus \scG$.

Observe that every tower as in Eq.~(\ref{eq:PauliTower}) determines a gauge group $\Phi \leq \scG \leq \scP$. Conversely, every gauge group $\Phi \leq \scG \leq \scP$ determines a tower as in Eq.~(\ref{eq:PauliTower}). Thus, in the usual formalism, a \emph{subsystem stabilizer code} is defined to be a subgroup $\Phi \leq \scG \leq \scP$. Moreover, one calls \(\scG\) a \emph{subsystem CSS code} if \(\scG\) admits a generating set that consists only of \(X\)-type or \(Z\)-type Pauli operators.

Any subsystem stabilizer code $\scG$ generates a von Neumann algebra \(\Alg \scG\) of linear operators on the physical Hilbert space $\dsH$. This algebra induces an orthogonal decomposition of $\dsH$ as \cite{P05, DKV24, H17}
\begin{equation}\label{eq:StabilizerSpaceDecomposition}
\dsH\cong \underset{\tau \scC(\scS)\in \scP/\scC(\scS)}{\bigoplus} \tau (\dsL\otimes \dsG),
\end{equation}
where \(\dsL\otimes \dsG\) is the subspace of $\dsH$ fixed by the stabilizer group $\scS$. With respect to this decomposition, the algebra $\Alg \scG$ takes the form
\begin{equation}\label{eq:StabilizerAlgebraDecomposition}
\Alg \scG \cong \underset{\tau \scC(\scS)\in \scP/\scC(\scS)}{\bigoplus} \tau {\parenth{\mathbbm{1}_{\dsL} \otimes \End(\dsG)}} \tau^{-1},
\end{equation}
where \(\End(\dsG)\) denotes the collection of linear operators on \(\dsG\). We call $\scG$ a \emph{subspace stabilizer code} if the gauge space $\dsG$ is one-dimensional.

\subsection{Subsystem stabilizer codes as subspaces of vector spaces}\label{sec:PrelimSSCVS}

In this work, we find it convenient to consider the tower in Eq.~(\ref{eq:PauliTower}) modulo phase factors, since this allows us to apply tools from linear algebra to analyze subsystem stabilizer codes. To begin, recall that $\scP / \Phi \cong G \times G$, so there exists some subspace $H \leq G \times G$ corresponding to the gauge group modulo phases, $\scG / \Phi \cong H$. Now, let $\scC(\scG)$ denote the centralizer of $\scG$. Then by Eq.~(\ref{eq:PauliPropertiesCommutation}), we have $\scC(\scG) / \Phi \cong H^\om$. Since $\generate{\Phi, \scS}$ is the center of $\scG$ and $\scC(\scS)$ is the centralizer of $\generate{\Phi, \scS}$, we have
\begin{equation}\label{eq:VSTower}
\begin{tikzpicture}
    [
    node distance = \nd and \nd,
    on grid,
    baseline = (current bounding box.center)
    ]

    \node (1) {\(\scP / \Phi\)};
    \node (2) [below=of 1] {\( \scC(\scS)/ \Phi\)};
    \node (3) [below=of 2] {\(\scG/ \Phi\)};
    \node (4) [below=of 3] {\( {\generate{\Phi,\scS}}/ \Phi\)};
    \node (5) [below=of 4] {\(\Phi/ \Phi\)};

    \node (eq1) [right=of 1] {\(\cong\)};
    \node (eq2) [right=of 2] {\(\cong\)};
    \node (eq3) [right=of 3] {\(\cong\)};
    \node (eq4) [right=of 4] {\(\cong\)};
    \node (eq5) [right=of 5] {\(\cong\)};

    \node (6) [right=of eq1]{\(G \times G\)};
    \node (7) [right=of eq2] {\(H + H^\omega\)};
    \node (8) [right=of eq3] {\(H\)};
    \node (9) [right=of eq4] {\(H \cap H^\omega\)};
    \node (10) [right=of eq5] {\(0\)};
    
    \draw
        (1) --
        (2) --
        (3) --
        (4) --
        (5);

    \path
        (1) -- (eq1) -- (6)
        (2) -- (eq2) -- (7)
        (3) -- (eq3) -- (8)
        (4) -- (eq4) -- (9)
        (5) -- (eq5) -- (10);

    \draw
        (6) --
        (7) --
        (8) --
        (9) --
        (10);
\end{tikzpicture}.
\end{equation}

We claim that this tower still encodes the parameters of the original subsystem stabilizer code. To see this, note that by the third and fourth isomorphism theorems \cite{Dummit_Foote_2004}, we have
\begin{align}
    \scC(\scS) / \scG &\cong {\parenth{H+H^
    \om}} / H \text{ and} \label{eq:logicalpauli}\\
    \scG / {\generate{\Phi,\scS}} &\cong H / {\parenth{H \cap H^\om}}, \label{eq:gaugepauli}
\end{align}
so the number of logical and gauge qudits can still be detected in Eq.~(\ref{eq:VSTower}). Moreover, note that the weight of a Pauli operator is phase invariant, so the code distance can still be detected in Eq.~(\ref{eq:VSTower}). This motivates the following proposition, which describes the parameters of a subsystem stabilizer code in the vector space setting (see also \cite[Theorem 5]{KS08}).

\begin{proposition}\label{prop:VSTowerParameters}
Let \(\om\) be an antisymmetric form on \(G\times G\), let \(H\leq G\times G\), and let $n \coloneqq \dim G$. Then
\begin{equation}\label{eq:HTower}
\begin{tikzpicture}
    [
    node distance = \nd and \nd,
    on grid,
    baseline = (current bounding box.center)
    ]

    \node (1) {\(G \times G\)};
    \node (2) [below=of 1] {\(H + H^\omega\)};
    \node (3) [below=of 2] {\(H\)};
    \node (4) [below=of 3] {\(H \cap H^\omega\)};
    \node (5) [below=of 4] {\(0\)};
    
    \draw
        (1) -- node [right] {\(n-k-r\)}
        (2) -- node [right] {\(2k\)}
        (3) -- node [right] {\(2r\)}
        (4) -- node [right] {\(n-k-r\)}
        (5);
\end{tikzpicture}
\end{equation}
for some \(k,r\in \dsZ\). Moreover, we can write
\begin{equation}
d = \min \swt\left({\parenth{H+H^{\om}}}\setminus H\right)
\end{equation}
for some $d \in \dsZ$.
\end{proposition}
\begin{proof}
See Appendix \ref{app:ProofVSTowerParameters}.
\end{proof}

This discussion leads to an alternative definition of subsystem stabilizer codes in the vector space setting.

\begin{definition}
A \emph{subsystem stabilizer code} is a subspace $H \leq G \times G$. We call \(H\) a \emph{subsystem CSS code} if \(H\) is a direct product of subspaces of $G$. With \(n,k,r,d\) as in Proposition \ref{prop:VSTowerParameters}, we say that \(H\) encodes \(k\) \emph{logical qudits} and \(r\) \emph{gauge qudits} into \(n\) \emph{physical qudits}, and that $H$ has \emph{distance} $d$. To summarize these parameters, we call \(H\) an \(\left[\left[n,k,r,d\right]\right]\) code.
\end{definition}

\begin{table*}[ht]
\centering
\footnotesize
% \begin{tabular}{ |c|c|c| } 
%  \hline
%  Name & Subgroup of $\scP$& Subspace of $G \times G$\\ 
%  \hline
%  Pauli group & $\scP$ & $G \times G$ \\ 
%  Normalizer & $\scC(\scS)$ & $H + H^\omega$ \\ 
%  Gauge group & $\scG$ & $H$ \\ 
%  Stabilizer group & $\left\langle \Phi, \scS \right\rangle$ & $H \cap H^\omega$ \\ 
%  Phases & $\Phi$ & $0$ \\ 
%  \hline
% \end{tabular}

\begin{tabular}{l|c|ccc}
\toprule 
Name & Subgroup of ${\cal P}$ & \multicolumn{3}{c}{Subspace of $G\times G$}\tabularnewline
 &  & Subsystem Stabilizer & Subsystem CSS & Subspace CSS \tabularnewline
\midrule
Centralizer & ${\cal C}({\cal S})/\Phi$ & $H+H^{\omega}$ & $H_{X}+H_{Z}^{\theta}\times H_{Z}+H_{X}^{\theta}$ & $H_{Z}^{\theta}\times H_{X}^{\theta}$\tabularnewline
Gauge & ${\cal G}/\Phi$ & $H$ & $H_{X}\times H_{Z}$ & $H_{X}\times H_{Z}$\tabularnewline
Stabilizer & ${\left\langle \Phi,{\cal S}\right\rangle} /\Phi$ & $H\cap H^{\omega}$ & $H_{X}\cap H_{Z}^{\theta}\times H_{Z}\cap H_{X}^{\theta}$ & $H_{X}\times H_{Z}$\tabularnewline
\midrule 
Logical Paulis & ${\cal C}({\cal S})/{\cal G}$ & ~~$(H+H^{\omega})/H$~~ & ~~{${(H_{X}+H_{Z}^{\theta})/H_{X}\times(H_{Z}+H_{X}^{\theta})/H_{Z}}$}~~ & ~~{${H_{Z}^{\theta}/H_{X}\times H_{X}^{\theta}/H_{Z}}$}~~\tabularnewline
Gauge Paulis & ${\cal G}/{{\left\langle \Phi,{\cal S}\right\rangle}} $ & $H/(H\cap H^{\omega})$ & {$H_{X}/(H_{X}\cap H_{Z}^{\theta})\times H_{Z}/(H_{Z}\cap H_{X}^{\theta})$} & {$0$}\tabularnewline
\midrule 
Gauge-preserving & ${\cal C}({\cal G}) / \Phi$ &  $H^\omega$ &  $H_{Z}^{\theta}\times H_{X}^{\theta}$ & $H_{Z}^{\theta}\times H_{X}^{\theta}$ \tabularnewline
Bare logical Paulis~~~ & ${\cal C}({\cal G})/{\left\langle \Phi,{\cal S}\right\rangle} $ &  $H^\omega / H\cap H^{\omega}$ &  $H_{Z}^{\theta} / {\parenth{H_{X}\cap H_{Z}^{\theta}}} \times {H_{X}^{\theta} / {\parenth{H_{Z}\cap H_{X}^{\theta}}}}$ & ${H_{Z}^{\theta}/H_{X}\times H_{X}^{\theta}/H_{Z}}$\tabularnewline
\bottomrule
\end{tabular}

% \caption{Dictionary between the names (left column), Pauli group representations (middle column), and vector space representations (right column) of five components of a subsystem stabilizer code. As is most evident in the right column, these five components can all be derived from one object, namely, the gauge group. See also Table 2 in \cite{KS22}.}
\caption{Dictionary between the names (first column), Pauli group representations (second column), and vector space representations of seven components of a subsystem stabilizer code (third column), subsystem CSS code (fourth column), and subspace CSS code (fifth column) (see also \cite[Table 2]{KS22}). For a subsystem stabilizer code, these components are derived from one object --- the gauge group \(H\) --- up to the phases of the stabilizer group elements. For a subsystem CSS code, these components are derived from two objects --- \(H_X\) and \(H_Z\) --- which are two unrelated classical linear codes. Imposing the relations in Eq.~(\ref{eq:css-subspace}) yields the subspace CSS construction.}

\label{tab:conversion}
\end{table*}

%%%%%%%%%%%%%%%%%%%%%%%%%%%%%%%%%%%%%%%%%%%%%
\section{Subsystem CSS codes}\label{sec:subsystemCSS}
%%%%%%%%%%%%%%%%%%%%%%%%%%%%%%%%%%%%%%%%%%%%%
% Table \ref{tab:conversion} summarizes the translation between the usual perspective and the vector space perspective on subsystem stabilizer codes.

In this section, we summarize what is known about subsystem CSS codes, and we fill
in some details that, to our knowledge, have been missing in the general case.

Subsystem CSS codes are constructed from two classical linear codes \(C_1\) and \(C_2\),
which are subspaces of the vector space $G=\mathbb{F}_{p}^{n}$. 
% A basis for each classical linear code can be arranged in the columns of the code's generator matrix, and a basis for the orthogonal complement (i.e., the dual) of each classical linear code can be arranged in the rows of the original code's parity-check matrix.
%In the above symplectic construction, $H_{Z}$ ($H_{X}$) is the parity-check matrix of $C_{1}$ ($C_{2}$).
The subsystem CSS construction does not require any relations between
the two classical codes.
Our construction is the same as that in Ref.~\cite[Corollary 4]{AKS06} via the association \(C_1\leftrightarrow H_{X}\) and \(C_2 \leftrightarrow H_{Z}\); we define the gauge group to be the direct product of the two classical codes, i.e.,
\begin{equation}
    H \coloneqq H_X \times H_Z.
\end{equation}

% reduces to that in Ref.~\cite[Thm.~7.3]{}\vva{doi:10.1002/9783527805785.ch1} when the codes \(H_{X}\) and \(H_{Z}\) are associated with the codes dual to \(C_2\) and \(C_1\), respectively.

% The columns of the generator matrices of $H_X$ and $H_Z$
% %rows of the parity-check matrices of $H_X$ and $H_Z$
% correspond to the $X$-type and $Z$-type gauge group generators, which yields the gauge group ${\cal G}/\Phi \cong H_{X}\times H_{Z}$ up to phases. 
Restricting to the subspace case requires additionally that the gauge group and stabilizer group coincide up to phases, which forces the gauge group to be abelian.
This imposes the relations
\begin{equation}\label{eq:css-subspace}
\begin{aligned}
    H_{X} \leq H_{Z}^{\theta} \quad \text{and} \quad H_{Z} \leq H_{X}^{\theta},
\end{aligned}
\end{equation}
which are the well-known relations between the two classical codes in the subspace CSS construction.

In general, $H_{X}$ does not have to be a subspace of $H_{Z}^{\theta}$, and the
overlap between the two spaces controls the parameters of the subsystem CSS code.
For example, suppose that we fix the dimensions of $H_{X}$ and $H_{Z}$,
but we choose $H_{Z}$ such that $H_{Z}^{\theta}$
has less overlap with $H_{X}$. Then the resulting subsystem
CSS code encodes more logical and gauge qudits into the same number
of physical qudits. This is demonstrated in the following proposition,
which expresses the tower in Eq.~(\ref{eq:HTower}) and the code
parameters in Proposition \ref{prop:VSTowerParameters} in terms of $H_{X}$ and $H_{Z}$.
% (see also \cite[Corollary 4]{AKS06}). 

% For a subsystem CSS code $H = H_X \times H_Z$, the tower in Eq.~(\ref{eq:HTower}) can be expressed in terms of the $X$-type and $Z$-type components $H_X$ and $H_Z$ (see also Corollary 4 in \cite{AKS06}).

\begin{proposition}\label{prop:CSSStructure}
Let \(H=H_X\times H_Z\leq G\times G\) be an \(\left[\left[n,k,r,d\right]\right]\) subsystem CSS code. Then
\begin{equation}\label{eq:HTowerCSS}
\begin{tikzpicture}
    [
    node distance = \nd and \nd,
    on grid,
    baseline = (current bounding box.center)
    ]

    \node (1) {\(G \times G\)};
    \node (2) [below=of 1] {\(H + H^\omega\)};
    \node (3) [below=of 2] {\(H\)};
    \node (4) [below=of 3] {\(H \cap H^\omega\)};
    \node (5) [below=of 4] {\(0\)};

    \node (eq1) [right=of 1] {\(=\)};
    \node (eq2) [right=of 2] {\(=\)};
    \node (eq3) [right=of 3] {\(=\)};
    \node (eq4) [right=of 4] {\(=\)};
    \node (eq5) [right=of 5] {\(=\)};

    \node (6) [right=of eq1]{\(G\)};
    \node (7) [right=of eq2] {\(H_X + H_Z^\theta\)};
    \node (8) [right=of eq3] {\(H_X\)};
    \node (9) [right=of eq4] {\(H_X \cap H_Z^\theta\)};
    \node (10) [right=of eq5] {\(0\)};

    \node (t1) [right=of 6] {\(\times\)};
    \node (t2) [right=of 7] {\(\times\)};
    \node (t3) [right=of 8] {\(\times\)};
    \node (t4) [right=of 9] {\(\times\)};
    \node (t5) [right=of 10] {\(\times\)};

    \node (11) [right=of t1]{\(G\)};
    \node (12) [right=of t2] {\(H_Z + H_X^\theta\)};
    \node (13) [right=of t3] {\(H_Z\)};
    \node (14) [right=of t4] {\(H_Z \cap H_X^\theta\)};
    \node (15) [right=of t5] {\(0\)};
    
    \draw
        (1) --
        (2) -- node [right] {\(2k\)}
        (3) -- node [right] {\(2r\)}
        (4) --
        (5);

    \path
        (1) -- (eq1) -- (6)
        (2) -- (eq2) -- (7)
        (3) -- (eq3) -- (8)
        (4) -- (eq4) -- (9)
        (5) -- (eq5) -- (10);

    \draw
        (6) --
        (7) -- node [right] {\(k\)}
        (8) -- node [right] {\(r\)}
        (9) --
        (10);

    \path
        (6) -- (t1) -- (11)
        (7) -- (t2) -- (12)
        (8) -- (t3) -- (13)
        (9) -- (t4) -- (14)
        (10) -- (t5) -- (15);

    \draw
        (11) --
        (12) -- node [right] {\(k\)}
        (13) -- node [right] {\(r\)}
        (14) --
        (15);
        
\end{tikzpicture}
\end{equation}
Moreover, let
\begin{equation}
d^{H_X} \coloneqq \min \wt\left({\parenth{H_X+H_Z^{\theta}}}\setminus H_X\right)
\end{equation}
be the minimum weight of a non-gauge $X$-type logical operator, and let
\begin{equation}d^{H_Z} \coloneqq \min \wt\left({\parenth{H_Z+H_X^{\theta}}}\setminus H_Z\right)\end{equation}
be the minimum weight of a non-gauge $Z$-type logical operator. Then
\begin{equation}\label{eq:SubsystemCSSD}
d = \min {\left\{ d^{H_X}, d^{H_Z} \right\}}.
\end{equation}
\end{proposition}
\begin{proof}
See Appendix \ref{app:ProofCSSStructure}.
\end{proof}
 
The above tower (Eq.~(\ref{eq:HTowerCSS})) is summarized in the fourth column of Table~\ref{tab:conversion},
while the fifth column shows the consequences of the subspace restriction (Eq.~(\ref{eq:css-subspace})). Each subspace in these columns is a direct product, where the second factor is obtained from the first factor by switching the letters $X$ and $Z$. This illustrates the $X\leftrightarrow Z$ symmetry of the CSS construction.

Recall that the subspace CSS construction yields a simple expression for a basis
of logical codewords in terms of cosets of the underlying classical
codes \cite{CS96,NC10,G10}. Indeed, the logical $X$-type Paulis commute and generate the entire codespace (up to scalars) when acting on the all-zero logical state, so their labels --- cosets of $H_{X}$ in $H_{Z}^{\theta}$ --- form a complete set of quantum numbers with which we can label a basis of logical codewords. Since a subsystem code can be interpreted as a subspace code with some of its logical qudits relegated to gauge qudits, we can extend this construction to subsystem CSS codes.

In the subsystem case, the logical and gauge qudits are labeled by
elements of their own subspaces, $(H_{X}+H_{Z}^{\theta})/H_{X}$ and $H_{X}/(H_{X}\cap H_{Z}^{\theta})$.
A basis for the codespace is given by the collection of vectors of the form
\begin{equation}
\left|l,g\right\rangle \coloneqq \frac{1}{\sqrt{{\abs{H_X \cap H_Z^\theta}}}} \sum_{s\in H_{X}\cap H_{Z}^{\theta}}\left|l+g+s\right\rangle,
\end{equation}
where $l$ and $g$ represent elements in $(H_{X}+H_{Z}^{\theta})/H_{X}$ and $H_{X}/(H_{X}\cap H_{Z}^{\theta})$, respectively 
%and where addition is done modulo the qudit dimension \(p\)
(see Appendix \ref{app:parameterize}). This reduces to the subspace CSS
construction upon imposing Eq.~(\ref{eq:css-subspace}).

%%%%%%%%%%%%%%%%%%%%%%%%%%%%%%%%%%%%%%%%%%%%%%%%%%%%%%%%%%%%%%%%%%%%%%%%%%%%%
\section{Stabilizer-to-CSS mapping}\label{sec:Mapping}

In this section, we show that every subsystem stabilizer code can be used to construct a subsystem CSS code with comparable parameters (see \cite[Theorem 1]{KP13} for the analogous result for subspace stabilizer codes). Our construction is simple to express in the notation of Sections \ref{sec:PrelimPauli} and \ref{sec:PrelimSSC}. Namely, if $\scG$ is a subsystem stabilizer code generated by $\set{X^{a_j}Z^{b_j}}_{j=1}^{2r}$, then the corresponding subsystem CSS code $\Delta(\scG)$ is generated by $\set{X^{a_j} \otimes X^{b_j}, Z^{b_j} \otimes Z^{-a_j}}_{j=1}^{2r}$ (see Examples \ref{ex:fivequbit}, \ref{ex:DoubleSemion}, and \ref{ex:Z1N}). However, to prove that the codes $\scG$ and $\Delta(\scG)$ do indeed have comparable parameters, it is convenient to work in the vector space formalism introduced in Sections \ref{sec:PrelimPauliVS} and \ref{sec:PrelimSSCVS}. Thus, we define our mapping $\Delta$ in the vector space setting, and we exhibit its key properties in the following lemma.

\begin{lemma}\label{lem:DeltaLemma}
For any $(a,b) \in G \times G$, define
\begin{equation}
    \Psi(a,b) \coloneqq (b,-a).
\end{equation}
For any $H \leq G \times G$, define
\begin{equation}\label{eq:OmegaHDefinition}\De (H)\coloneqq H\times \Psi (H) \leq (G \times G) \times (G \times G).
\end{equation}
Then for any $H,K \leq G \times G$, we have
\begin{subequations}\label{eq:deltaprops}
    \begin{align}
 \De (H+K) &=\De (H)+\De (K), \label{eq:OmegaJoin} \\
 \De (H\cap K) &=\De (H)\cap \De (K), \label{eq:OmegaMeet} \\
 \De (H^{\om}) &=\De (H)^{\om}, \text { and} \label{eq:OmegaComplement} \\
 \dim \De (H) &=2 \dim H. \label{eq:OmegaDimension}
\end{align}
\end{subequations}
That is, $\Delta$ is a lattice embedding that respects $\om$-complement and doubles dimension.
\end{lemma}
\begin{proof}
To verify Eq.~(\ref{eq:OmegaComplement}), observe that for any \((a,b),(c,d)\in G\times G\), we have
\begin{subequations}
\begin{align}
 \theta (\Psi (a,b),\Psi (c,d))&=\theta ((a,b),(c,d)),\\ \om (\Psi (a,b),\Psi (c,d))&=\om ((a,b),(c,d)), \\ \theta (\Psi (a,b),(c,d))&=\om ((a,b),(c,d)), \text{ and}\\ \om((a,b),\Psi (c,d))&=\theta ((a,b),(c,d)).
\end{align}
\end{subequations}
Then for any subspace \(H\leq G\times G\), we have
\begin{subequations}\label{eq:psiprop}
\begin{align}
 \Psi (H^{\theta})&=\Psi (H)^{\theta}, \label{eq:PsiProperties1}\\
 \Psi (H^\om)&=\Psi (H)^{\om}, \label{eq:PsiProperties2}\\
 H^{\om}&=\Psi (H)^{\theta}, \text{ and}\label{eq:PsiProperties3}\\
 H^{\theta} &=\Psi(H)^{\om}.\label{eq:PsiProperties4}
\end{align}
\end{subequations}
Thus, we have
\begin{equation}
\begin{aligned}
 \De (H^\om)&=H^{\om} \times \Psi (H^\om)\\
 &=\Psi (H)^{\theta} \times H^{\theta} \\
 &=(H\times \Psi (H))^{\om}\\
 &=\De (H)^{\om},
\end{aligned}
\end{equation}
where the second line holds by Eq.~(\ref{eq:psiprop}), and the third line holds by Eq.~(\ref{eq:DirectProductOmega}). To verify Eq.~(\ref{eq:OmegaJoin}), observe that
\begin{equation}
\begin{aligned}
 \De (H)+\De (K)&=H\times \Psi (H)+K\times \Psi (K)\\
 &=(H+K)\times (\Psi (H)+\Psi (K))\\
 &=(H+K)\times \Psi (H+K)\\
 &=\De (H+K).
\end{aligned}
\end{equation}
To verify Eq.~(\ref{eq:OmegaMeet}), observe that
\begin{equation}
\begin{aligned}
 \De (H\cap K)&=\De \! \left(\left(H^{\om} +K^{\om}\right)^{\om}\right)\\
 &=\left(\De (H)^{\om}+\De (K)^{\om}\right)^{\om}\\
 &=\De(H)\cap \De (K).
\end{aligned}
\end{equation}
Finally, to verify Eq.~(\ref{eq:OmegaDimension}), note that $\Psi$ is an isomorphism on $G \times G$, and take the dimension of the right-hand side of Eq.~(\ref{eq:OmegaHDefinition}).
\end{proof}

We are now ready to state and prove the main result of this section.

\begin{theorem}\label{thm:mapping}
Let \(H\leq G\times G\) be an \(\left[\left[n,k,r,d\right]\right]\) subsystem stabilizer code. Then $\Delta(H)$ is a \(\left[\left[2n,2k,2r,d'\right]\right]\) subsystem CSS code, where $d \leq d' \leq 2d$. Moreover, if \(H\) admits a collection of generators with symplectic weight at most $w$, then \(\De (H)\) admits a collection of generators with symplectic weight at most $2w$.
\end{theorem}
\begin{proof}
Apply $\Delta$ to the tower
\begin{equation}\begin{tikzpicture}
    [
    node distance = \nd and \nd,
    on grid,
    baseline = (current bounding box.center)
    ]

    \node (1) {\(G \times G\)};
    \node (2) [below=of 1] {\(H + H^\omega\)};
    \node (3) [below=of 2] {\(H\)};
    \node (4) [below=of 3] {\(H \cap H^\omega\)};
    \node (5) [below=of 4] {\(0\)};
    
    \draw
        (1) -- node [right] {\(n-k-r\)}
        (2) -- node [right] {\(2k\)}
        (3) -- node [right] {\(2r\)}
        (4) -- node [right] {\(n-k-r\)}
        (5);
\end{tikzpicture}\end{equation}
to obtain the tower
\begin{equation}\label{eq:DeltaTower}
\begin{tikzpicture}
    [
    node distance = \nd and \nd,
    on grid,
    baseline = (current bounding box.center)
    ]

    \node (1) {\((G\times G) \times (G\times G)\)};
    \node (2) [below=of 1] {\(\De (H) + \De (H)^\omega\)};
    \node (3) [below=of 2] {\(\De (H)\)};
    \node (4) [below=of 3] {\(\De (H) \cap \De (H)^\omega\)};
    \node (5) [below=of 4] {\(0\)};
    
    \draw
        (1) -- node [right] {\(2n-2k-2r\)}
        (2) -- node [right] {\(4k\)}
        (3) -- node [right] {\(4r\)}
        (4) -- node [right] {\(2n-2k-2r\)}
        (5);
\end{tikzpicture},\end{equation}
which shows that $\Delta(H)$ is a \(\left[\left[2n,2k,2r,d'\right]\right]\) subsystem CSS code.

To bound $d'$, first note that by Proposition \ref{prop:CSSStructure} along with Eq.~(\ref{eq:psiprop}), it follows that the tower in Eq.~(\ref{eq:DeltaTower}) is equal to the tower
\begin{equation}\label{eq:DeltaTowerSimplified}\begin{tikzpicture}
    [
    node distance = \nd and \nd,
    on grid,
    baseline = (current bounding box.center)
    ]

    \node (1) {\(G \times G\)};
    \node (2) [below=of 1] {\(H + H^\omega\)};
    \node (3) [below=of 2] {\(H\)};
    \node (4) [below=of 3] {\(H \cap H^\omega\)};
    \node (5) [below=of 4] {\(0\)};

    \node (eq1) [right=of 1] {\(\times\)};
    \node (eq2) [right=of 2] {\(\times\)};
    \node (eq3) [right=of 3] {\(\times\)};
    \node (eq4) [right=of 4] {\(\times\)};
    \node (eq5) [right=of 5] {\(\times\)};

    \node (6) [right=of eq1]{\(\Psi(G \times G)\)};
    \node (7) [right=of eq2] {\(\Psi(H + H^\omega)\)};
    \node (8) [right=of eq3] {\(\Psi(H)\)};
    \node (9) [right=of eq4] {\(\Psi(H \cap H^\omega)\)};
    \node (10) [right=of eq5] {\(\Psi(0)\)};
    
    \draw
        (1) --
        (2) --
        (3) --
        (4) --
        (5);

    \path
        (1) -- (eq1) -- (6)
        (2) -- (eq2) -- (7)
        (3) -- (eq3) -- (8)
        (4) -- (eq4) -- (9)
        (5) -- (eq5) -- (10);

    \draw
        (6) --
        (7) --
        (8) --
        (9) --
        (10);
\end{tikzpicture}.
\end{equation}
Since the map $\Psi$ is weight-preserving, it follows from Proposition \ref{prop:CSSStructure} that
\begin{equation}
d' = \min \wt \left({\parenth{H+H^{\om}}}\setminus H\right).
\end{equation}
Now, for any $x \in G \times G$, we have
\begin{equation}
    \swt(x) \leq \wt(x) \leq 2\swt(x).
\end{equation}
Minimizing over ${\parenth{H+H^{\om}}}\setminus H$, we find that
\begin{equation}
    d \leq d' \leq 2d,
\end{equation}
as needed.

Finally, if \(T\) is a generating set for \(H\), then \(\left(T\times 0\right)\cup \left(0\times
\Psi (T)\right)\) is a generating set for \(\De (H)\), and for any $t \in T$, we have
\begin{align}
    \swt (t, 0) &\leq 2 \swt(t) \text{ and}\\
    \swt (0, \Psi(t)) &\leq 2 \swt(\Psi(t)) = 2 \swt(t),
\end{align}
and we are done.
\end{proof}

Intuitively, our mapping in Eq.~(\ref{eq:OmegaHDefinition}) embeds two copies of the subsystem stabilizer code $H$ into the subsystem CSS code $\De(H)$. One copy $H \times 0$ consists entirely of $X$-type Pauli operators, while the other (isomorphic) copy $0 \times \Psi(H)$ consists entirely of $Z$-type Pauli operators. This is shown explicitly in Eq.~(\ref{eq:DeltaTowerSimplified}), which explains why $\Delta$ can be called a ``doubling" mapping.

We remark that Theorem \ref{thm:mapping} can be generalized to subsystem stabilizer codes over modular qudits (see Appendix \ref{app:Generalize}). Now, we illustrate our construction with three examples.

\begin{example}[Five-qubit code]\label{ex:fivequbit}
Consider the $[[5,1,0,3]]$ code \cite{LMPZ96} \eczoo{stab_5_1_3}, which is the smallest subspace stabilizer code that can correct any single-qubit Pauli error. Note that the $[[5,1,0,3]]$ code does not admit a CSS representation under single-qubit Clifford rotations. The stabilizer group is
\begin{align}\label{eq:five-qubit-stabs}
    \mathcal G = \langle ZXXZI, IZXXZ, ZIZXX, XZIZX \rangle,
\end{align}
and the logical Pauli operators are
\begin{equation}
    \begin{aligned}
        \bar X &= XXXXX \text{ and} \\ \bar Z &= ZZZZZ.
\end{aligned}
\end{equation}
Applying the mapping in Lemma \ref{lem:DeltaLemma} and Theorem \ref{thm:mapping}, the corresponding subspace CSS code has ten physical qubits with stabilizer group
\begin{align}
\begin{split}
      \Delta (\mathcal G) = \bigg \langle \ &\begin{matrix}XIIXI\\IXXII\end{matrix}, \ \begin{matrix}IXIIX\\IIXXI\end{matrix}, \ \begin{matrix}XIXII\\IIIXX\end{matrix}, \ \begin{matrix}IXIXI\\XIIIX\end{matrix}, \\
    &\begin{matrix}IZZII\\ZIIZI\end{matrix}, \ \begin{matrix}IIZZI\\IZIIZ\end{matrix}, \ \begin{matrix}IIIZZ\\ZIZII\end{matrix}, \ \begin{matrix}ZIIIZ\\IZIZI\end{matrix} \ \bigg \rangle,
\end{split}
\end{align}
where we arrange the ten physical qubits in two rows of five for readability.
\end{example}

\begin{example}[Double semion code]\label{ex:DoubleSemion}
The double semion code is a topological subspace stabilizer code \cite{Ellison22} \eczoo{double_semion} encoding one logical qubit. To define the code, we place one $4$-dimensional qudit on each edge of a square lattice with periodic boundary conditions. Then, we specify the gauge group, which is generated by the four Pauli operators
\begin{equation}
\label{eq:doublesemiongenerators}
\raisebox{-.5\height}{\includegraphics[scale=0.4]{DS.pdf}}
\end{equation}
per unit cell, where
\begin{align}
    Y &\coloneqq e^{i\pi/4}X^\dagger Z^\dagger \text{ and} \\
    \tilde Y &\coloneqq e^{i\pi/4}Z^\dagger X.
\end{align}

Note that the double semion code is not a subspace CSS code. In fact, one can argue that there does not exist a shallow depth Clifford circuit transforming the double semion code into a subspace CSS code. Indeed, if such a circuit exists, then one can define a commuting projector Hamiltonian which is stoquastic \footnote{the $X$ stabilizers in Eq.~(\ref{eq:doublesemiongenerators}) only contain positive off-diagonal entries, while the $Z$ stabilizers are diagonal and can therefore be shifted by a constant such that all coefficients are positive}, and whose ground state lies in the double semion topological phase. However, it is known that the double semion phase has an intrinsic sign problem~\cite{Hastings16,Smith20}.

Although the double semion code is not Clifford equivalent to a subspace CSS code, one can apply the mapping in Lemma \ref{lem:DeltaLemma} and Theorem \ref{thm:mapping} to construct a subspace CSS code with parameters comparable to that of the double semion code. The resulting subspace CSS code has two physical qudits per edge and eight stabilizer generators per unit cell given by
\begin{equation}
\raisebox{-.5\height}{\includegraphics[scale=0.4]{DS_CSS.pdf}}.
\end{equation}
Here, we use the shorthand $X^a X^b$ to mean $X^a \otimes X^b$ (as opposed to operator multiplication) and similarly for the $Z$-type operators.

It would be interesting to determine the phase of matter corresponding to the subspace CSS code above. This code cannot be two copies of the double semion code due to the sign problem, but we conjecture that this code is equivalent to two copies of a $\mathbb Z_2$ toric code. We leave further investigation of this point to future work.

\end{example}

\begin{example}[$\mathbb Z_N^{(1)}$ subsystem code]
\label{ex:Z1N}
The $\mathbb Z_N^{(1)}$ subsystem code is a topological subsystem stabilizer code based on the qudit generalization of the Kitaev honeycomb model introduced in~\cite{Barkeshli15,Ellison23} \eczoo{qudit_znone}. It encodes one logical qudit, which is of dimension $N$ if $N$ is odd and $\frac{N}{2}$ if $N$ is even. To define the code, we place one $N$-dimensional qudit on each vertex of a honeycomb lattice. The gauge group is generated by the two-body checks
\begin{equation}
\raisebox{-.5\height}{\includegraphics[scale=1]{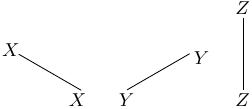}}
\end{equation}
where
\begin{equation}
    Y \coloneqq X Z.
\end{equation}

Since each vertex hosts three checks of $XX$, $YY$, and $ZZ$ type, there does not exist a shallow depth Clifford circuit transforming the $\mathbb Z_N^{(1)}$ subsystem code into a subsystem CSS code. However, using the mapping in Lemma \ref{lem:DeltaLemma} and Theorem \ref{thm:mapping}, we can construct a subsystem CSS code with comparable parameters. This code is generated by the six checks
\begin{equation}
\raisebox{-.5\height}{\includegraphics[scale=1]{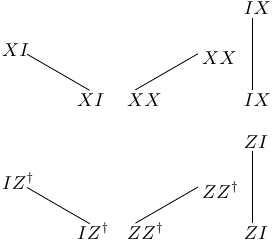}}
\end{equation}
per unit cell.
\end{example}

The proof of Theorem \ref{thm:mapping} suggests a general strategy for constructing subsystem CSS codes from subsystem stabilizer codes. Indeed, any function $\De$ satisfying Eq.~(\ref{eq:deltaprops}) yields a mapping of codes that doubles the number of physical, logical, and gauge qudits. Moreover, by modifying Eq.~(\ref{eq:OmegaDimension}) appropriately, one can allow $\De$ to map to subspaces of $G^m \times G^m$ for any $m \in \dsN$, thus producing $[[m n,m k,m r,d']]$ subsystem CSS codes as output. We hope that the proof of Theorem \ref{thm:mapping} and the related linear algebraic ideas will inspire new mappings from subsystem stabilizer codes to subsystem CSS codes.

We remark that our Theorem \ref{thm:mapping} is a direct generalization of Theorem 1 in \cite{KP13}. In addition, results analogous to our Theorem \ref{thm:mapping} have been obtained previously \cite{BTL10, LAV22, GNK14}. The mappings in these works take $[[n,k,0,d]]$ subspace stabilizer codes to $[[4n,2k,0,d']]$ subspace CSS codes. These mappings take subspace stabilizer codes to Majorana fermion codes to subspace CSS codes, which increases the complexity of the proof (a mapping from $[[n,k,r,d]]$ subsystem stabilizer codes to $[[4n,2k,2r,2d]]$ subsystem CSS codes was obtained in \cite[Theorem 5.7]{B11}, but the proof remains rather long despite avoiding the intermediate Majorana fermion codes). Nonetheless, the mappings in \cite{BTL10, LAV22, GNK14} can be formulated as maps analogous to $\De$ in Lemma \ref{lem:DeltaLemma}. With some modifications, it is possible to demonstrate the validity of the mappings in \cite{BTL10, LAV22, GNK14} by following the general proof structure in Lemma \ref{lem:DeltaLemma} and Theorem \ref{thm:mapping}.

\section{Recovery procedure for subsystem CSS codes}\label{sec:CSSRecovery}

Having shown that subsystem CSS codes achieve performance comparable to that of subsystem stabilizer codes, we now study subsystem CSS codes in greater detail. In this section, we present a recovery procedure to correct an unknown Pauli error affecting the physical Hilbert space of a subsystem CSS code. Our recovery procedure generalizes the usual Steane-type decoder for subspace stabilizer codes \cite{CS96,NC10,G10} to the subsystem setting.

To begin, let \(H=H_X\times H_Z\leq G\times G\) be a subsystem CSS code. Let $\dsL \otimes \dsG \leq \dsH$ be the codespace of $H$, i.e., the fixed space of the stabilizer group of $H$. Suppose that a code state
\begin{equation}
    \ket{\psi} = \ket{\psi_\dsL} \otimes \ket{\psi_\dsG} \in \dsL \otimes \dsG
\end{equation}
is corrupted by a Pauli error \(X^aZ^b \in \scP\), where \((a,b)\in G\times G\). Our recovery procedure identifies the error $(a,b)$ up to some unknown gauge term in $H_X \times H_Z$.

The first step in our recovery procedure is purely quantum: in Section \ref{sec:CSSRecoverySyndrome}, we use gauge generator measurements to determine the syndrome of the Pauli error. The second step in our recovery procedure is purely classical: in Section \ref{sec:CSSRecoveryClassical}, we use two classical linear codes to identify the Pauli error up to gauge terms. Our discussion focuses on correcting the $X$-component of the error; correcting the $Z$-component is similar.

\subsection{Quantum step: syndrome measurement}\label{sec:CSSRecoverySyndrome}

% In the first step of our recovery procedure, 
We determine the syndrome of the Pauli error by measuring gauge generators. Usually, the syndrome of a Pauli error is defined to be its commutation with each of the stabilizer generators. In the linear algebraic framework, the syndrome of a Pauli error is defined (for subsystem CSS codes) as follows.

\begin{definition}\label{def:SynXandSynZ}
Let \(H=H_X\times H_Z \leq G \times G\) be a subsystem CSS code. Let $(a,b) \in G \times G$. Let
\begin{equation}
\begin{aligned}
\Syn_X: G &\to G/{\parenth{H_X + H_Z^\theta}} \\
a &\mapsto a + \parenth{H_X + H_Z^\theta}
\end{aligned}
\end{equation}
and
\begin{equation}
\begin{aligned}
\Syn_Z: G &\to G/{\parenth{H_Z + H_X^\theta}} \\
b &\mapsto b + \parenth{H_Z + H_X^\theta}
\end{aligned}
\end{equation}
be the canonical projections. We call $\Syn_X(a)$ the \emph{syndrome} of $a$ and similarly for $b$.
\end{definition}

For intuition, note that the space $H_X + H_Z^\theta$ is precisely the space of all $X$-type logical operators, i.e., the space of all $X$-type Pauli operators that commute with all the $Z$-type stabilizers in $H_Z\cap H_X^{\theta}$. Thus, the syndrome of an $X$-type error $a$ is precisely the set of all $X$-type Pauli operators with the same $Z$-type stabilizer commutations as $a$.

Now, our goal is to use gauge generator measurements to determine the syndrome of $a$. To begin, fix bases for the spaces
\begin{equation}
    \begin{array}{cc}
        H_X = \generate{h_X} & H_Z = \generate{h_Z} \\
        H_X\cap H_Z^{\theta} = \generate{{s_X}} & H_Z\cap H_X^{\theta} = \generate{{s_Z}}
    \end{array}.
\end{equation}
Observe that for each basis vector ${s_Z}$ in $H_Z\cap H_X^{\theta}$, we have
\begin{equation}{s_Z}=\sum_{h_Z} c_{h_Z}h_Z\end{equation}
for some \(c_{h_Z}\in \mathbb{F}\), so
\begin{equation}\label{eq:GaugeProductToStabilizer}
Z^{{s_Z}}=\prod _{h_Z} \left(Z^{h_Z}\right)^{c_{h_Z}}.
\end{equation}
That is, every \(Z\)-type stabilizer generator \(Z^{{s_Z}}\) can be written as a product of \(Z\)-type gauge generators \(Z^{h_Z}\). This simple observation suggests that one can determine the syndrome of the error $a$, i.e., the commutation of $X^a$ with the stabilizer generators $Z^{s_Z}$, by measuring the gauge generators $Z^{h_Z}$. To see this formally, note that the stabilizer generators \(X^{{s_X}}\), \(Z^{{s_Z}}\) and the \(Z\)-type gauge generators \(Z^{h_Z}\) all commute, so there exists an orthonormal basis \(\left\{\ket{e_i} \right\}_i\) for \(\dsH\) in which the \(X^{{s_X}}\), \(Z^{{s_Z}}\), and \(Z^{h_Z}\) are all diagonal. Write the corrupted state in this basis, i.e.,
\begin{equation}\label{eq:CorruptedStateExpanded}
X^aZ^b\ket{\psi} = \sum_i \psi_i \ket{e_i},
\end{equation}
where the $\psi_i \in \dsC \setminus 0$. Since $\ket{\psi} \in \dsL \otimes \dsG$, by applying $Z^{{s_Z}}$ to both sides above, we find that
\begin{equation}Z^{{s_Z}}\ket{e_i} = e^{\frac{2\pi i}{p}\theta ({s_Z},a)} \ket{e_i}\end{equation}
for each $\ket{e_i}$ in Eq.~(\ref{eq:CorruptedStateExpanded}). Similarly, applying $X^{{s_X}}$ to both sides in Eq.~(\ref{eq:CorruptedStateExpanded}), we find that
\begin{equation}X^{{s_X}}\ket{e_i} = e^{-\frac{2\pi i}{p}\theta ({s_X},b)} \ket{e_i}\end{equation}
for each $\ket{e_i}$ in Eq.~(\ref{eq:CorruptedStateExpanded}). Now, since the $\ket{e_i}$ are also eigenstates of the $Z^{h_Z}$, measuring any $Z^{h_Z}$ projects the corrupted state $X^aZ^b\ket{\psi}$ onto a subspace spanned by some subset of the $\ket{e_i}$ in Eq.~(\ref{eq:CorruptedStateExpanded}). Thus, by Eq.~(\ref{eq:GaugeProductToStabilizer}), if we measure each gauge generator \(Z^{h_Z}\) and multiply each corresponding measurement outcome $c_{h_Z}$ times, we obtain the value $e^{\frac{2\pi i}{p}\theta ({s_Z},a)}$, from which we can determine $\theta({s_Z}, a)$. Repeating this for each ${s_Z}$, we can determine the image of $a$ under the linear map $a \mapsto \sum_{{s_Z}} \theta ({s_Z},a) {s_Z}$ on $G$. Since this map has kernel ${\parenth{H_Z\cap H_X^{\theta}}}^\theta = H_X+H_Z^{\theta}$, we can determine $a+\parenth{H_X+H_Z^{\theta}} = \Syn_X(a)$, as desired.

We claim that this procedure does not disturb the logical component of the corrupted state. Intuitively, this is because the measurement of gauge generators can only modify the gauge component of the corrupted state. To see this formally, recall that the gauge generators \(Z^{h_Z}\) are of the form
\begin{equation}
Z^{h_Z} \cong \underset{\tau}{\bigoplus} \, {\tau \parenth{\mathbbm{1}_{\dsL} \otimes U_{\dsG}^\tau} \tau^{-1}}
\end{equation}
for some unitaries $U_{\dsG}^\tau \in \End(\dsG)$. Since any projection onto an eigenspace of \(Z^{h_Z}\) can be written as a polynomial in \(Z^{h_Z}\), any measurement of \(Z^{h_Z}\) effects a projection of the form
\begin{equation}
P \cong \underset{\tau}{\bigoplus} \, {\tau \parenth{\mathbbm{1}_{\dsL} \otimes P_{\dsG}^\tau} \tau^{-1}}
\end{equation}
on \(X^aZ^b{\parenth{\ket{\psi _{\mathbb{L}}} \otimes\ket{\psi _{\mathbb{G}}}}}\), where the $P_{\dsG}^\tau \in \End({\dsG})$ are some projections. Thus, after measuring the gauge generators \(Z^{h_Z}\), the corrupted state is of the form \(X^aZ^b{\parenth{\ket{\psi _{\mathbb{L}}} \otimes\ket{\psi _{\mathbb{G}}'}}}\) for some gauge state \(\ket{\psi _{\mathbb{G}}'} \in \mathbb{G}\).

Finally, since the state $\ket{\psi _{\mathbb{L}}} \otimes\ket{\psi _{\mathbb{G}}'}$ still belongs to the codespace $\dsL \otimes \dsG$, one can apply an analogous procedure to determine $\Syn_Z(b)$ by measuring the $X$-type gauge generators $X^{h_X}$, and after this procedure, the state of the system is \(X^aZ^b{\parenth{\ket{\psi _{\mathbb{L}}} \otimes\ket{\psi _{\mathbb{G}}''}}}\) for some \(\ket{\psi _{\mathbb{G}}''} \in \mathbb{G}\). This concludes the first step in our recover procedure.

We remark that our syndrome measurement procedure uses homogeneous scheduling of gauge generator measurements. That is, in our procedure, all $Z$-type gauge generators are measured, and subsequently all $X$-type gauge generators are measured. However, other measurement schedules might be relevant for fault-tolerant implementations of our recovery procedure; this is related to the notion of ``gauge fixing'' \cite{HB21, SS23}. Indeed, the special case of our recovery procedure for subspace CSS codes (i.e., the Steane recovery procedure) is fault-tolerant \cite{G10}. Future work could investigate potential fault-tolerant implementations of our generalized recovery procedure for subsystem CSS codes.

\subsection{Classical step: error recovery}\label{sec:CSSRecoveryClassical}

In the first step of our recovery procedure, we determined the syndrome $\Syn_X(a)$ of the error $a$. That is, we determined the image of $a$ under the map
\begin{equation}
\begin{aligned}
\Syn_X: G \to G/{\parenth{H_X + H_Z^\theta}}.
\end{aligned}
\end{equation}
In the second step of our recovery procedure, we use $\Syn_X$ as the parity check of the classical linear code $H_X + H_Z^\theta$ to identify the error $a$ up to gauge terms in $H_X$. We emphasize here a key point: the syndrome map used in the first step of our recovery procedure to determine the commutation of the error with the stabilizers is \emph{precisely the same as} the parity check of the classical linear code used in the second step of our recovery procedure to determine the error up to gauge terms. Thus, in our recovery procedure, the quantum and classical syndromes coincide, which further illustrates the connection between subsystem CSS codes and classical linear codes.

To begin, we introduce a general notion of a classical linear code capable of identifying an error up to some prescribed ``redundant" subcode.

\begin{definition}\label{def:ClassicalLinearCodeProperties}
Let $N$ be a vector space over $\dsF_p$. A \emph{classical linear code} is a subspace $K \leq N$. A \emph{parity check} for $K$ is a linear map \(F:N\to M\) such that $\ker F = K$. A \emph{redundant subcode} for $K$ is a subspace $R \leq K$. We say that $K$ has \emph{logical dimension} \(k\coloneqq \dim K\) and \emph{physical dimension} \(n\coloneqq \dim N\). The \emph{distance of $K$ with respect to $R$} is
\begin{equation}d^R\coloneqq \min \wt(K \setminus R).\end{equation}
To summarize these parameters, we call $K$ an \(\left[n,k,d^R\right]\) code.
\end{definition}

Now, suppose that we only wish to identify an error $a \in N$ up to some redundant subcode $R$. That is, suppose that we only wish to recover the coset $a+R$ from the syndrome $F a$ of the error. As shown in the following proposition, this is possible when $a+R$ contains some vector of sufficiently small weight relative to the code distance $d^R$.

\begin{proposition}\label{prop:ClassicalLinearCodeCorrectability}
Let $K$ be a classical linear code with parity check \(F:N\to M\) and redundant subcode $R$. Let \(a\in N\) be unknown, but suppose that $F a$ is known. Suppose that \(\min \wt(a + R) < \frac{d^R}{2}\). Then we can determine \(a+R\), i.e., we can determine $a+k$ for some unknown redundant vector $k \in R$.
\end{proposition}
\begin{proof}
See Appendix \ref{app:proofblah}.
\end{proof}

The special case of Proposition \ref{prop:ClassicalLinearCodeCorrectability} with trivial redundant subcode is well known \cite{HKS21}. We state it as a corollary below for a sanity check.

\begin{corollary}
Let $K$ be a classical linear code with parity check \(F:N\to M\). Let \(a\in N\) be unknown, but suppose that $F a$ is known. Suppose that \(\wt(a) < \frac{\min \wt K}{2}\). Then we can determine $a$.
\end{corollary}

% \vva{What do you mean by \(\dim(H_X + H_Z^\theta)\), the span of rows or columns? Its is the ROWS that form the classical code whose generator matrix is \(H_X + H_Z^\theta\). I think we should specify which classical code we mean here. Upstairs ure using \(H_{X,Z}\) as parity-check matrices, i.e., their rows are a basis for the dual code to that used to define the CSS construcoin. I checked this and it reproduces the right thing.}
Let us return to the second step of our recovery procedure. Suppose that $H = H_X \times H_Z$ is an \(\left[\left[n,k,r,d\right]\right]\) subsystem CSS code. Write $k_X \coloneqq \dim {\parenth{H_X + H_Z^\theta}}$ and $k_Z \coloneqq \dim {\parenth{H_Z + H_X^\theta}}$ (note that $k_X + k_Z = n + k + r$ by Proposition \ref{prop:VSTowerParameters}). Let $d^{H_X},d^{H_Z}$ be as in Proposition \ref{prop:CSSStructure}. One verifies from the definitions that $H_X + H_Z^\theta$ with parity check $\Syn_X$ and redundant subcode $H_X$ is an $[n, k_X,d^{H_X}]$ classical linear code, and $H_Z + H_X^\theta$ with parity check $\Syn_Z$ and redundant subcode $H_Z$ is an $[n, k_Z,d^{H_Z}]$ classical linear code.

Now, suppose that the error $a$ satisfies $\min \wt(a+H_X) < \frac{d^{H_X}}{2}$. Then by Proposition \ref{prop:ClassicalLinearCodeCorrectability}, we can recover the error class $a+H_X$ from the syndrome $\Syn_X(a)$, which was determined previously in the first step of our recovery procedure. To correct the $X$-type error, we can pick any $a+h_X \in a+H_X$ and apply  ${\parenth{X^{a+h_X}}}^{-1}$ to the corrupted code state \(X^aZ^b{\parenth{\ket{\psi _{\mathbb{L}}} \otimes\ket{\psi _{\mathbb{G}}''}}}\), which yields the state  \(Z^b{\parenth{\ket{\psi _{\mathbb{L}}} \otimes\ket{\psi _{\mathbb{G}}'''}}}\). Assuming that the error $b$ satisfies $\min \wt(b+H_Z) < \frac{d^{H_Z}}{2}$, we can apply a similar procedure to correct the $Z$-type error and restore the state \(\ket{\psi _{\mathbb{L}}} \otimes\ket{\psi _{\mathbb{G}}''''}\). This concludes the second step in our recovery procedure.

We remark that the second step in our recovery procedure performs as well as can be expected, given that the distance of the parent subsystem CSS code is $d = \min\{d^{H_X},d^{H_Z}\}$ (as shown in Proposition \ref{prop:CSSStructure}). Indeed, if the original Pauli error $X^a Z^b$ satisfies $\swt(a,b) < \frac{d}{2}$, which is the usual correctability condition for subsystem stabilizer codes, then $a$ and $b$ satisfy $\wt(a) < \frac{d^{H_X}}{2}$ and $\wt(b) < \frac{d^{H_Z}}{2}$, so by the above discussion, the second step in our recovery procedure succeeds.

One might note that the second step in our recovery procedure identifies a little more than just the coset $a+H_X$. In fact, as indicated in the proof of Proposition \ref{prop:ClassicalLinearCodeCorrectability}, our recovery procedure always identifies a representative $a+h_X$ of $a+H_X$ with weight less than $\frac{d^{H_X}}{2}$ (provided that such a representative exists). One might ask whether a weaker recovery procedure exists, which only recovers the coset $a+H_X$ and nothing else. We explore this idea further in Appendix \ref{app:Par}.

% Finally, recall that in the usual Steane recovery procedure for subspace CSS codes \cite{CS96,NC10,G10}, one can construct a basis for the codespace that is parameterized by the classical linear codes $H_X$ and $H_Z$ underlying the subspace CSS code $H$. This parameterization can be generalized to the subsystem case; we provide more details in Appendix \ref{app:parameterize}.

\section{Structure of subsystem stabilizer codes} \label{sec:SubsystemStabilizerCode}
After a detailed study of subsystem CSS codes, we now turn our attention to the structure of general subsystem stabilizer codes. Our main observation is that every subsystem stabilizer code is associated with two subsystem CSS codes and an isomorphism that determine the structure of the subsystem stabilizer code. In Section \ref{sec:Goursat}, we formalize this observation using Goursat's Lemma. Then, in Section \ref{sec:SSCStructure}, we investigate the structure of subsystem stabilizer codes by associating two subsystem CSS codes to each subsystem stabilizer code in Eq.~(\ref{eq:HTower}), just as we associated two classical linear codes to each subsystem CSS code in Eq.~(\ref{eq:HTowerCSS}). Finally, in Section \ref{sec:SSCGeneralize}, we apply this theory to introduce two generalizations of subsystem CSS codes.

\subsection{Goursat data of the gauge group}\label{sec:Goursat}

We can quantify the degree of ``non-CSS-ness'' of any subsystem stabilizer code $H\leq G\times G$ by sandwiching it between two subsystem CSS codes,
\begin{equation}
N\leq H\leq E,
\end{equation}
where the \emph{internal CSS code} $N$ is the largest subsystem CSS code
contained in $H$, and the \emph{external CSS code} $E$ is the
smallest subsystem CSS code containing $H$.

The internal CSS code consists of all elements in the gauge group that are purely
$X$-type or purely $Z$-type, i.e., all $X^a$ and $Z^b$ where $X^a, Z^b \in \scG$. In particular, the internal CSS code equals $H$ iff $H$ is itself a subsystem CSS code, and the internal CSS code is trivial iff $H$ has no pure $X$ or pure $Z$ elements.

The external CSS code is constructed by independently peeling off the $X$- and $Z$-components of each element in the gauge group. That is, the external CSS code contains all $X^a$ and $Z^b$ such that $X^a Z^b \in \scG$. In particular, the external CSS code equals $H$ iff $H$ is itself a subsystem CSS code, and the external CSS code always contains $H$.
%For example, for the element $XZZXI$, the external code group gains two elements, $XIIXI$ and $IZZII$.
%In particular, the external CSS code contains both the internal
%code and the gauge group, since any element of $H$ can be obtained
%by multiplying two corresponding elements of the external code. 

If $H$ is not a subsystem CSS code, the inclusions
\begin{equation}
N < H < E
\end{equation}
are strict. In particular, the internal code is a strict subspace of
the external code, since there are elements of $E$ that cannot be
obtained by taking products of the pure gauge group elements
in $N$. The relative size of the internal and external codes corresponds to the number of non-CSS elements in the gauge group. This claim is made precise using Goursat's Lemma.

Goursat's Lemma associates a gauge group with a set of \emph{Goursat
data} --- the internal code $N=N_{X}\times N_{Z}$, the external code $E=E_{X}\times E_{Z}$, and an isomorphism $\phi:E_{X}/N_{X} \to E_{Z}/N_{Z}$ that associates the \(X\)- and \(Z\)-components of each gauge group element $(e_X,e_Z)\in H$ up to pure $X,Z$ elements:
\begin{equation}
e_X+N_{X}\stackrel[\phi^{-1}]{\phi}{\rightleftharpoons}e_Z+N_{Z}.
\end{equation}
% for any $(e_X,e_Z)\in H$. the number of cosets is the same on both sides since for any $(e_X,e_Z)\in H$,
% $e_X$ is an element of $N_{X}$ if and only if $e_Z$ is an element of $N_{Z}$. Indeed, if $(e_X,0)$ is a pure $X$ term in
% the gauge group, then $(0,e_Z)=(e_X,e_Z)-(e_X,0)$ is a pure $Z$ term in the gauge group, and visa versa.
In other words, for each gauge group element $(e_X,e_Z)$, the isomorphism $\phi$ associates the non-pure component of \(e_X\) with the non-pure component of its partner \(e_Z\).
The number of such independent pairs quantifies the ``non-CSS-ness'' of the code: if $E_{X}/N_{X}$ has one element, then $H$ is CSS, but if $E_{X}/N_{X}$ is large, then $H$ is highly ``non-CSS''. This is all stated formally below.
% We apply Goursat's Lemma formally below.
% Goursat's Lemma addresses a very natural question: given a vector space $G$, what do the subspaces of $G \times G$ look like? Intuitively, Goursat's Lemma reads as follows. Given a subsystem stabilizer code $H\leq G\times G$, one can consider the largest subsystem CSS code $N_X \times N_Z$ contained in $H$, as well as the smallest subsystem CSS code $E_X \times E_Z$ containing $H$. Then $H$ consists of ``blocks" of $N_X \times N_Z$ inside $E_X \times E_Z$ (see Fig. \ref{fig:Figure1} for an illustration). This is stated formally below.

\begin{lemma}[Goursat's Lemma] \label{lem:Goursat}
Let \(G\) be a vector space.
\begin{enumerate}[label=(\arabic*),ref=(\arabic*)]
\item \label{lem:GoursatA} Let \(H \leq G\times G\). Define
\begin{subequations}\label{eq:defsr}\begin{align}
{E_X} &\coloneqq \{{e_X}\in G \mid \exists {e_Z}\in G, ({e_X},{e_Z})\in H\}, \label{eq:DefSX}\\
{E_Z} &\coloneqq \{{e_Z}\in G \mid  \exists {e_X}\in G, ({e_X},{e_Z})\in H\}, \\
{N_X} &\coloneqq \{{n_X}\in G \mid  ({n_X},0)\in H\}, \text{ and} \\
{N_Z} &\coloneqq \{{n_Z}\in G \mid  (0,{n_Z})\in H\}.\label{eq:DefRZ}
\end{align}\end{subequations}
For each ${e_X} \in {E_X}$, define
\begin{equation}
\phi ({e_X}+{N_X}) \coloneqq {e_Z}+{N_Z},
\end{equation}
where $({e_X},{e_Z}) \in H$. Then
\begin{equation}\begin{array}{c}
G \times G \\
| \\
 {E_X} \times {E_Z}\\
 | \\
 {N_X} \times {N_Z}\\
\end{array}\end{equation}
and
\begin{equation}{E_X}/{N_X}\overset{\phi} {\xrightarrow{\sim}}{E_Z}/{N_Z}.\end{equation}

\item \label{lem:GoursatB} Conversely, let
\begin{equation}\begin{array}{c}
G \times G \\
| \\
 {E_X} \times {E_Z}\\
 | \\
 {N_X} \times {N_Z}\\
\end{array}\end{equation}
and
\begin{equation}{E_X}/{N_X}\overset{\phi} {\xrightarrow{\sim}}{E_Z}/{N_Z}.\end{equation}
Define
\begin{equation}\label{eq:HFromGoursat}
\begin{aligned}
    H\coloneqq \{({e_X},{e_Z})\in {E_X}&\times {E_Z} \, | \\ \phi ({e_X}+{N_X})&={e_Z}+{N_Z}\}.
\end{aligned}
\end{equation}
Then $H\leq G\times G$.

\item \label{lem:GoursatC} The constructions in \ref{lem:GoursatA} and \ref{lem:GoursatB} are inverses.
\end{enumerate}
\end{lemma}
\begin{proof}
This is \cite[Theorem 3]{MG22}. See also \cite[Theorem 4]{AC09}.
\end{proof}

We emphasize that item \ref{lem:GoursatC} above states that the correspondence between \(H\) and its Goursat data is bijective. Thus, we can write
\begin{equation}\label{eq:goursatcorresp}
\goursatCorrespondence{G}{H}{{E_X}}{{E_Z}}{{N_X}}{{N_Z}}{\phi}
\end{equation}
to indicate that the subsystem stabilizer code on the left-hand side corresponds via Lemma \ref{lem:Goursat} to the Goursat data on the right-hand side. This allows us to study the structure of a subsystem stabilizer code $H$ by studying the corresponding internal code $N$, external code $E$, and isomorphism $\phi$.

%The isomorphism $\phi$ allows for associating $H$ with one or more ``blocks" (i.e., cosets) of $N_X \times N_Z$ inside $E_X \times E_Z$.
Goursat's Lemma implicitly states that the subsystem stabilizer code $H$ consists of ``blocks" (i.e., cosets) of $N_X \times N_Z$ inside $E_X \times E_Z$. The isomorphism $\phi$ essentially pairs pure $X$ elements $e_X$ with pure $Z$ elements $e_Z$ to obtain the representatives $(e_X, e_Z)$ of these cosets comprising $H$. We illustrate this in Fig.~\ref{fig:Figure1} and state it formally below.

\begin{proposition}\label{prop:GoursatProperties}
Suppose that Eq.~(\ref{eq:goursatcorresp}) holds. Fix a basis
\begin{equation}
    \left\{{e_X^j}+{N_X}\right\}_j
\end{equation}
for \({E_X}/{N_X}\), and consider the basis
\begin{equation}
    \left\{\phi \left({e_X^j}+{N_X}\right)\eqqcolon {e_Z^j}+{N_Z}\right\}_j
\end{equation}
for \({E_Z}/{N_Z}\). Then
\begin{equation}\label{eq:GoursatHExplicit}
\begin{aligned}
H=\Bigg\{\sum_j &c_j(e_X^j, e_Z^j) + (n_X, n_Z)\in {E_X}\times {E_Z} \, \Bigg| \, \\
&c_j\in \dsF_p,(n_X, n_Z)\in {N_X \times N_Z}\Bigg\}.
\end{aligned}
\end{equation}
Moreover, \(H\) is a subsystem CSS code iff \({E_X}={N_X}\), which occurs iff \({E_Z}={N_Z}\).
\end{proposition}
\begin{proof}
See Appendix \ref{app:ProofGoursatProperties}.
\end{proof}

\begin{figure}[htbp]
\includegraphics[width = 0.48 \textwidth]{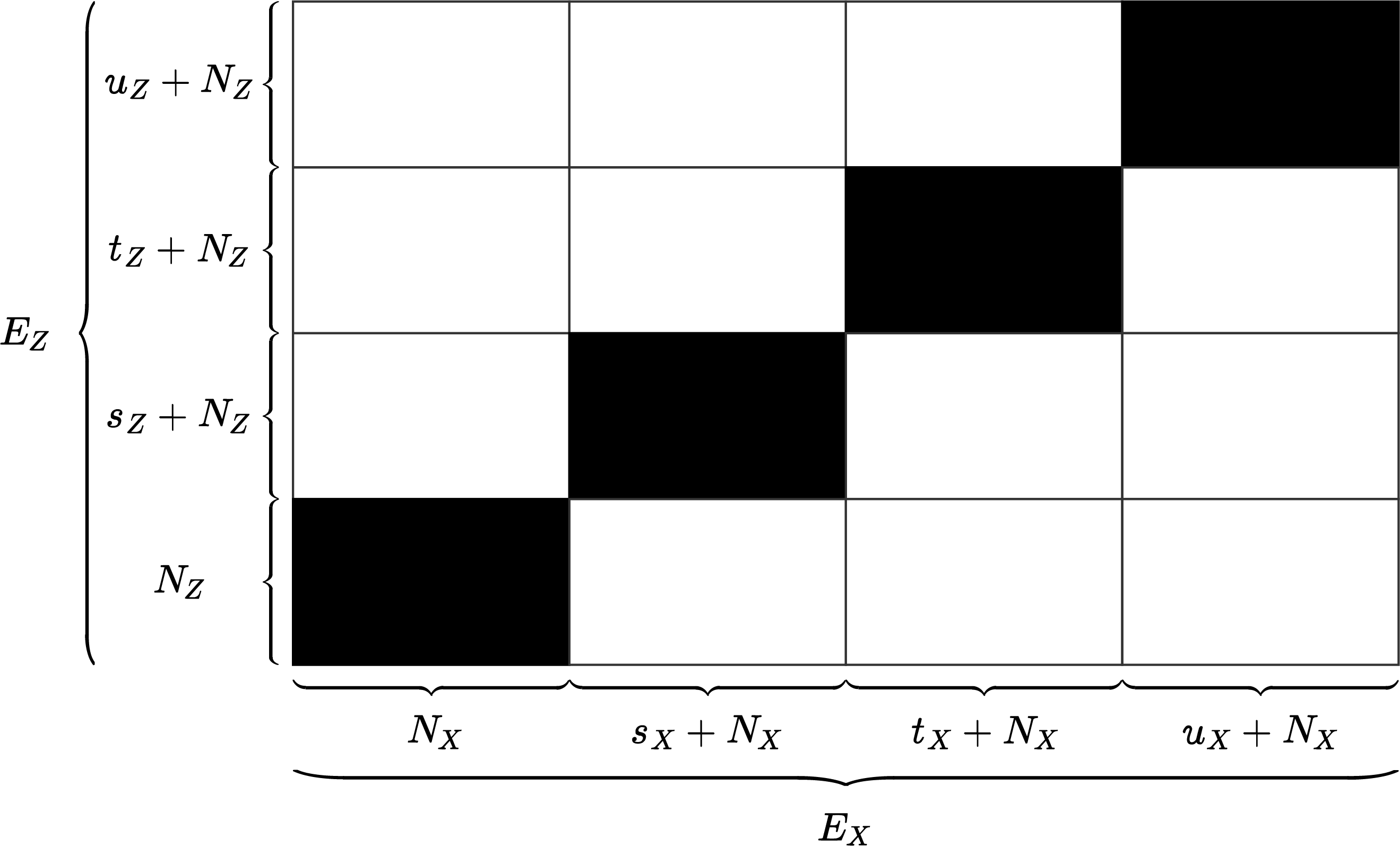}
\caption{Illustration of Lemma \ref{lem:Goursat} and Proposition \ref{prop:GoursatProperties} with $|{E_X}/{N_X}|= 4 =|{E_Z}/{N_Z}|$. Any subspace \(H \leq G \times G\) consists of ``blocks'' --- i.e., cosets --- (dark shaded rectangles) of a direct product \({N_X}\times {N_Z}\) (bottom left dark shaded rectangle) contained inside a direct product \({E_X}\times {E_Z}\) (large rectangle). Shaded regions belong to $H$, while unshaded regions do not. Here, $\phi \left(v_X + {N_X}\right)= v_Z+{N_Z}$ for $v = s,t,u$.}
\label{fig:Figure1}
\end{figure}

As indicated in Eq.~(\ref{eq:GoursatHExplicit}), the isomorphism $\phi$ specifies which $X$-type operators ${e_X^j}$ and $Z$-type operators ${e_Z^j}$ are ``paired together'' to form generators of $H$. We illustrate this in the following example.

\begin{example}[5-qubit code]
The five-qubit code \cite{LMPZ96} is a stabilizer code, so its gauge group is the join of its stabilizer group and phases, \(\mathcal{G} = \langle \mathcal{S}, \Phi\rangle\).

Inspecting the stabilizer generators of this code in Eq.~(\ref{eq:five-qubit-stabs}), we see that there are no elements of the gauge group that are purely $X$-type or $Z$-type. Thus, the internal code $N_X \times N_Z$ is trivial.

On the other hand, the external code $E_X \times E_Z$ is generated by splitting the gauge group into its $X$-part and $Z$-part. That is, with a slight abuse of notation, we have
\begin{align}
{E_X} &= \langle IXXII, IIXXI, IIIXX, XIIIX\rangle \text{ and}\\
{E_Z} &= \langle ZIIZI, IZIIZ, ZIZII, IZIZI \rangle.
\end{align}

The isomorphism $\phi$ is constructed from combinations of ${e_X} \in {E_X}$ and ${e_Z} \in {E_Z}$ such that ${e_X} {e_Z} \in \mathcal G$. Thus, again with a slight abuse of notation, the action of $\phi$ is
\begin{align}
    IXXII &\mapsto ZIIZI,& IIXXI &\mapsto IZIIZ,\\
     IIIXX &\mapsto ZIZII,& XIIIX &\mapsto IZIZI.
\end{align}

An alternative set of generators for the external code is
\begin{align}
{E_X} &= \langle XXIII, IXXII, IIXXI, IIIXX\rangle \text{ and}\\
{E_Z} &= \langle ZZIII, IZZII, IIZZI, IIIZZ \rangle.
\end{align}
The external code can be gauged fixed to obtain the \([[5,1,0,3]]\) code. However, the external code has distance one, since $XIIII$ and $ZIIII$ are dressed logicals.
\end{example}

To aid in computer implementations, we note that the Goursat data of a subsystem stabilizer code $H$ can be expressed explicitly in terms of its generators. This is useful for explicitly computing the Goursat data of subsystem stabilizer codes.

\begin{proposition}\label{prop:ExplicitGoursat}
Suppose that
\begin{equation}
    H =  \left\langle \left({e_X^1},{e_Z^1}\right),\ldots ,\left({e_X^l},{e_Z^l}\right)\right\rangle \leq G\times G.
\end{equation}
Define the $\dim G \times l$ matrices
\begin{equation}
    \pi_X\coloneqq \left(
\begin{array}{ccc}
 | & & | \\
 {e_X^1} & \ldots & {e_X^l} \\
 | & & | \\
\end{array}
\right)
\end{equation}
and
\begin{equation}
    \pi_Z\coloneqq \left(
\begin{array}{ccc}
 | & & | \\
 {e_Z^1} & \ldots & {e_Z^l} \\
 | & & | \\
\end{array}
\right).
\end{equation}
Then
\begin{equation}\label{eq:ExplicitGoursat}
\goursatCorrespondence{G}{H}{\pi_X{\parenth{\dsF_p^l}}}{\pi_Z{\parenth{\dsF_p^l}}}{\pi_X{\parenth{\ker \pi_Z}}}{\pi_Z{\parenth{\ker \pi_X}}}{\phi}\end{equation}
for some isomorphism \(\phi\). Moreover, \(H\) is a subsystem CSS code iff
\begin{equation}\label{eq:KernelSum}
\ker \pi_X+\ker \pi_Z=\dsF_p^l.\end{equation}
\end{proposition}
\begin{proof}
See Appendix \ref{app:ProofExplicitGoursat}.
\end{proof}

Note that Eq.~(\ref{eq:KernelSum}) provides a characterization of subsystem CSS codes that can be readily checked given the generators of $H$. Indeed, to check whether a given subsystem stabilizer code is a subsystem CSS code, one can compute spanning sets for $\ker \pi_X$ and $\ker \pi_Z$, take their union to form a spanning set for $\ker \pi_X+\ker \pi_Z$, and compare the dimension of that spanning set with $l$. However, this procedure is more complicated than simply computing the reduced row echelon form of the generator matrix of $H$, so Eq.~(\ref{eq:KernelSum}) may be most interesting as a theoretical characterization of subsystem CSS codes.

\subsection{Goursat data of the stabilizer group}\label{sec:SSCStructure}

The stabilizer group \(H\cap H^{\om}\) of a subsystem stabilizer code \(H\) can be described in terms of the Goursat data of $H$. More broadly, our goal in this section is to determine the Goursat data of each subsystem stabilizer code in Eq.~(\ref{eq:HTower}) in terms of the Goursat data of $H$, just as we determined the classical linear codes comprising each subsystem CSS code in Eq.~(\ref{eq:HTowerCSS}) in terms of the classical linear codes comprising $H$. 

We first derive the Goursat data of the gauge-preserving group $H^\om$ in terms of that of $H$.

\begin{proposition}\label{prop:GoursatComplement}
Suppose that
\begin{equation}\goursatCorrespondence{G}{H}{{E_X}}{{E_Z}}{{N_X}}{{N_Z}}{\phi}.\end{equation}
Then there exist isomorphisms $\phi_X, \phi_Z$ in the diagram
\begin{equation}{N_Z^\theta} /{E_Z^\theta} \overset{\phi_Z}{\xrightarrow{\sim}}{E_Z}/{N_Z}\overset{\phi ^{-1}}{\xrightarrow{\sim}}{E_X}/{N_X}\overset{\phi_X}{\xleftarrow{\sim}}{N_X^\theta}/{E_X^\theta},\end{equation}
such that with
\begin{equation}\phi ^{\om} \coloneqq \phi_X^{-1}\phi ^{-1}\phi_Z,\end{equation}
we have
\begin{equation}\label{eq:abcd123}\goursatCorrespondence{G}{H^\om}{N_Z^\theta}{N_X^\theta}{E_Z^\theta}{E_X^\theta}{\phi^\om}\end{equation}
\end{proposition}
\begin{proof}
See Appendix \ref{app:ProofGoursatComplement}.
\end{proof}

Intuitively, Proposition \ref{prop:GoursatComplement} states that if a subsystem stabilizer code $H$ corresponds to the subsystem CSS codes $N \leq H \leq E$, then the complement $H^\om$ corresponds to the subsystem CSS codes $E^\om \leq H^\om \leq N^\om$ (c.f. Eq.~(\ref{eq:DirectProductOmega})).

Next, we derive the Goursat data of an intersection of codes $H \cap \tilde{H}$ in terms of the Goursat data of the original codes $H$ and $\tilde{H}$. This result is slightly more general than is needed to determine the Goursat data of the stabilizer $H \cap H^\omega$, since \(\tilde H\) does not have to equal \(H^\omega\).

\begin{proposition}\label{prop:GoursatIntersection}
Suppose that
\begin{equation}\goursatCorrespondence{G}{H}{{E_X}}{{E_Z}}{{N_X}}{{N_Z}}{\phi}\end{equation}
and
\begin{equation}\goursatCorrespondence{G}{\tilde{H}}{\tilde{{E_X}}}{\tilde{{E_Z}}}{\tilde{{N_X}}}{\tilde{{N_Z}}}{\tilde{\phi}}.\end{equation}

Then there exist subspaces
\begin{widetext}
\begin{equation}\label{eq:complicated_intersection}
\begin{array}{ccccccc}
 {E_X}\cap \tilde{{E_X}} & & \left({E_X}/{N_X}\right)\times \left(\tilde{{E_X}}/\tilde{{N_X}}\right) & & \left({E_Z}/{N_Z}\right)\times \left(\tilde{{E_Z}}/\tilde{{N_Z}}\right)
& & {E_Z}\cap \tilde{{E_Z}} \\
 | & & | & & | & & | \\
 T & , & \scU & , & \scV & , & W \\
 | & & | & & | & & | \\
 {N_X}\cap \tilde{{N_X}} & & 0 & & 0 & & {N_Z}\cap \tilde{{N_Z}} \\
\end{array}
\end{equation}
\end{widetext}
and isomorphisms $\phi_T, \phi_W$ in the diagram
\begin{equation}
 T\left/\left({N_X}\cap \tilde{{N_X}}\right)\right. \overset{\phi_T}{\xrightarrow{\sim}} \scU \overset{\phi \times \tilde{\phi}}{\xrightarrow{\sim}} \scV \overset{\phi_W}{\xleftarrow{\sim}} W\left/\left({N_Z}\cap \tilde{{N_Z}}\right)\right.
 \end{equation}
such that with
\begin{equation}\phi ^{\cap} \coloneqq \phi_W^{-1}\left(\phi \times \tilde{\phi} \right)\phi_T,\end{equation}
we have
\begin{equation}\goursatCorrespondence{G}{H\cap \tilde{H}}{T}{W}{\parenth{{N_X}\cap \tilde{{N_X}}}}{\parenth{{N_Z}\cap \tilde{{N_Z}}}}{\phi ^{\cap}}.\end{equation}
\end{proposition}
\begin{proof}
See Appendix \ref{app:ProofGoursatIntersection}.
\end{proof}

Intuitively, Proposition \ref{prop:GoursatIntersection} states that the internal code of the intersection of $H$ and $\tilde{H}$ is exactly the intersection of the internal codes of $H$ and $\tilde H$. However, as indicated in the first and fourth columns of Eq.~(\ref{eq:complicated_intersection}), the external code of the intersection of $H$ and $\tilde{H}$ need only be contained in (but not necessarily equal to) the intersection of the external codes of $H$ and $\tilde{H}$.

Propositions \ref{prop:GoursatComplement} and \ref{prop:GoursatIntersection} are sufficient to determine the Goursat data of each subsystem stabilizer code in Eq.~(\ref{eq:HTower}) in terms of the Goursat data of $H$. However, these expressions for subsystem stabilizer codes are not as explicit as the expressions for subsystem CSS codes in Eq.~(\ref{eq:HTowerCSS}). Nonetheless, Propositions \ref{prop:GoursatComplement} and \ref{prop:GoursatIntersection} might be useful for some applications, such as a characterizaztion of subsystem stabilizer codes.

\subsection{Characterizing subsystem stabilizer codes using the Goursat data of the stabilizer group}\label{sec:SSCGeneralize}

% \vva*{
% The number of blocks in the block structure induced by a code's Goursat data (see Fig. \ref{fig:Figure1}) quantifies ``how far'' a given gauge group is from being a CSS-type group.
% Having only one block means that the gauge group is of CSS type, but, in general, the Goursat data of of 
% }
Here, we characterize subsystem stabilizer codes according to the Goursat data of their stabilizer groups (see Fig.~\ref{fig:Figure2} for a summary).

%%%%%%%%%%%%%%%%%%%%%%%%%%%%%%%%%%%%%%%%%%%
\begin{figure}[htbp]
\includegraphics[width = 0.5\textwidth]{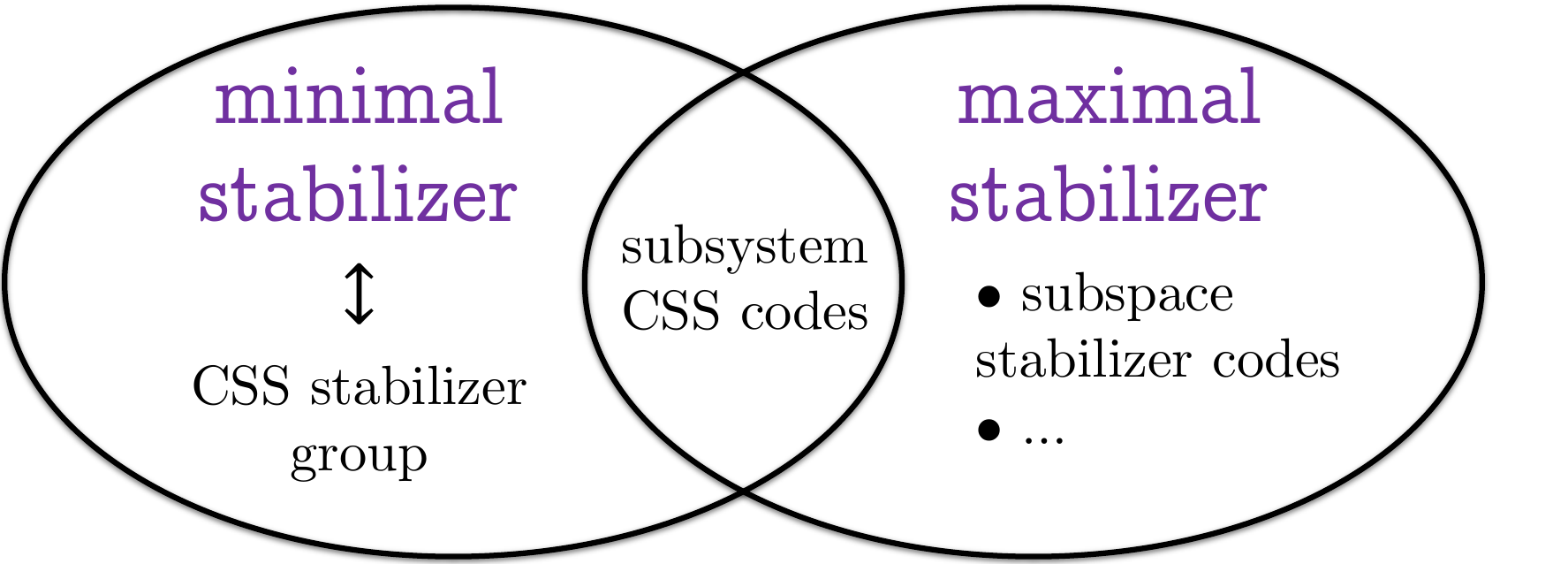}
\caption{Venn diagram depicting features of subsystem stabilizer codes with maximal and minimal stabilizer groups, per Definition \ref{def:maxmincenter}. A code has minimal stabilizer iff its stabilizer group is CSS. A code has minimal and maximal stabilizer iff its gauge group is CSS. Every subspace stabilizer code has maximal stabilizer, but there are codes with maximal stabilizer that are not subspace stabilizer codes.}
\label{fig:Figure2}
\end{figure}
%%%%%%%%%%%%%%%%%%%%%%%%%%%%%%%%%%%%%%%%%%%

Recall that a subsystem stabilizer code \(H\) can equivalently be described by its Goursat data --- the internal CSS code $N$, the external CSS code $E$, and the isomorphism \(\phi\). Here, we are interested in the number of blocks in the stabilizer group $H \cap H^\omega$ of $H$ (see Figure \ref{fig:Figure1}). To motivate this inquiry, recall that the stabilizer group has one block iff it is CSS, and the stabilizer group has many blocks iff it is ``far from CSS''. As we will show, the number of blocks in the stabilizer group is bounded above by a function of $N$ and $E$. This upper bound is independent of the particular isomorphism $\phi$, which determines the true number of blocks in the stabilizer group of $H$. Thus, for fixed internal and external codes $N$ and $E$, we can define two extreme classes of codes which, in some sense, achieve the largest and the smallest number of blocks admitted by $N$ and $E$.

%The internal and external codes set an upper bound on the number of blocks allowed by stabilizer groups %of all \(H\) with the same internal and external codes.
%The isomorphism \(\phi\) then determines the actual number of blocks admitted by a given code's stabilizer.

%When the upper bound is one, all codes admitting that particular \(N\) and \(E\) must be CSS subsystem codes.
%In all other cases, code families with such \(N,E\) can be further characterized by their number of blocks.
%Naturally, we can define the two extreme cases of codes with the smallest and the largest number of possible blocks.

% We now introduce two generalizations of subsystem CSS codes that may be of further interest. 
% These generalizations correspond to subsystem stabilizer codes whose stabilizer groups have the maximal and minimal number of blocks allowed in their block structure.

\begin{definition}\label{def:maxmincenter}
Let
\begin{equation}\goursatCorrespondence{G}{H}{{E_X}}{{E_Z}}{{N_X}}{{N_Z}}{\phi}.\end{equation}
By Propositions \ref{prop:GoursatComplement} and \ref{prop:GoursatIntersection}, we have
\begin{equation}\label{eq:macdonal}
\begin{array}{ccccc}
 & & \parenth{{E_X}\cap {N_Z^\theta}} \times \parenth{{E_Z}\cap {N_X^\theta}} & &  \\
 G\times G & & | & & \\
 | & \longleftrightarrow & T \times W &,& \phi ^{\cap} \\
 H\cap H^{\om} & & | & & \\
 & & \parenth{{N_X}\cap {E_Z^\theta}} \times \parenth{{N_Z}\cap {E_X^\theta}} & & 
\end{array}
\end{equation}
for some subspace \(T \times W\) (interpreted as the external code of the stabilizer) and some isomorphism $\phi^\cap$. We say that $H$ has \emph{maximal stabilizer} if
\begin{equation}
T \times W = \parenth{{E_X}\cap {N_Z^\theta}} \times \parenth{{E_Z}\cap {N_X^\theta}},
\end{equation}
i.e., if $T \times W$ is as large as possible. We say that \(H\) has \emph{minimal stabilizer} if
\begin{equation}
T \times W = \parenth{{N_X}\cap {E_Z^\theta}} \times \parenth{{N_Z}\cap {E_X^\theta}},
\end{equation}
i.e., if $T \times W$ is as small as possible.
\end{definition}

It is in the following sense that subsystem stabilizer codes with maximal or minimal stabilizer generalize subsystem CSS codes.

\begin{proposition}\label{prop:minimalmaximal}
A subsystem stabilizer code has both maximal and minimal stabilizer iff it is a subsystem CSS code.
\end{proposition}
\begin{proof}
See Appendix \ref{app:proofminimalmaximal}.
\end{proof}

Since subsystem CSS codes have many nice properties, as exhibited in this work and others, it may be interesting to study these two generalizations of subsystem CSS codes. In particular, the class of codes with minimal stabilizer is precisely the class of codes with CSS stabilizer, so it seems particularly natural to investigate whether codes with minimal stabilizer retain some useful properties enjoyed by subsystem CSS codes. In addition, the class of codes with maximal stabilizer is precisely the class of codes where the external code of the intersection of $H$ and $H^\omega$ is equal to the intersection of the external codes of $H$ and $H^\omega$ (see the discussion following Proposition \ref{prop:GoursatIntersection}), so this class of codes may also exhibit some interesting properties.

We remark that codes with minimal stabilizer have multiple characterizations.

\begin{corollary}\label{cor:DirectProductEquivalence}
The following are equivalent:
  \begin{enumerate}[label=(\arabic*),ref=(\arabic*)]
    \item\label{eq:Equiv0}
    $H$ has minimal stabilizer
    \item\label{eq:Equiv1}
    $H + H^{\omega}$ is a direct product
    \item\label{eq:Equiv2}
    $H \cap H^{\omega}$ is a direct product
    \item\label{eq:Equiv3}
    $H + H^{\omega} = ({E_X} + {N_Z^\theta}) \times ({E_Z} + {N_X^\theta})$
    \item\label{eq:Equiv4}
    $H \cap H^{\omega} = ({N_X} \cap {E_Z^\theta}) \times ({N_Z} \cap {E_X^\theta})$
  \end{enumerate}
\end{corollary}
\begin{proof}
By definition, \ref{eq:Equiv0} and \ref{eq:Equiv2} are equivalent. By Eq.~(\ref{eq:DirectProductOmega}), \ref{eq:Equiv1} and \ref{eq:Equiv2} are equivalent, and \ref{eq:Equiv3} and \ref{eq:Equiv4} are equivalent. By Propositions \ref{prop:GoursatComplement} and \ref{prop:GoursatIntersection}, \ref{eq:Equiv2} and \ref{eq:Equiv4} are equivalent.
\end{proof}

This corollary yields a simple algorithm to test whether a subsystem stabilizer code has minimal stabilizer. Indeed, given a spanning set for $H$, one can compute a spanning set for $H+H^\om$ and apply Eq.~(\ref{eq:KernelSum}) (or row reduction) to test if $H + H^\om$ is a direct product, i.e., if $H$ has minimal stabilizer. We leave a more detailed investigation of codes with maximal and minimal stabilizer for future research.

\section{Conclusion} \label{sec:Conclusion}

In this work, we study subsystem stabilizer codes in a linear algebraic framework by viewing Pauli groups as vector spaces and subsystem stabilizer codes as subspaces via the symplectic representation. 
This perspective allows us to construct a doubling map that converts any modular-qudit subsystem stabilizer code into a subsystem CSS code with comparable parameters (see Theorem \ref{thm:mapping} and Appendix \ref{app:Generalize}). 

The linear algebraic perspective also elucidates the correspondence between a subsystem CSS code and a pair of classical linear codes, yielding an explicit basis of codewords (see Section \ref{sec:subsystemCSS}), a generalized Steane-type recovery procedure (see Section \ref{sec:CSSRecovery}), and other results (see Appendices \ref{app:Par} and \ref{app:parameterize}). 

Finally, the linear-algebraic perspective enables a characterization of subsystem stabilizer codes via Goursat's Lemma, which yields two generalizations of subsystem CSS codes (see Section \ref{sec:SubsystemStabilizerCode}).

We conclude with questions and outlook for future investigations.
\begin{enumerate}
    \item Is our ``doubling'' map from \([[n,k,r,d]]\) subsystem stabilizer codes to \([[2n,2k,2r, d \leq d' \leq 2d]]\) subsystem CSS codes in Theorem \ref{thm:mapping} tight in terms of the required number of physical qudits?
    \item Which subsystem stabilizer codes yield ``doubled'' subsystem CSS codes with \textit{exactly} twice the code distance?
    \item Does our recovery procedure in Section \ref{sec:CSSRecovery} interact nicely with the ``doubled" subsystem CSS code?
    \item 
    Keeping in mind that Steane error correction is fault-tolerant for subspace codes, can our subsystem syndrome extraction procedure in Section \ref{sec:CSSRecoverySyndrome} be made fault-tolerant?
    Fault-tolerant Steane error correction for subspace CSS codes relies on the codes' transversal implementations of inter-block logical CNOT gates and check operator measurements. 
    A subsystem CSS code can be viewed as a subspace CSS code with some logical qubits relegated as gauge qubits.
    Thus, subsystem CSS codes inherit the ingredients required to make Steane error correction fault-tolerant \cite{G24}. 
    % Possible complications include that the stabilizer generators of physically relevant subsystem CSS codes may not have low weight.
    Proving fault tolerance of Steane error correction for subsystem CSS codes may be possible with, e.g., the exRec formalism \cite{AGP06}, and we leave this as an exciting question for future work.
    \item Does the Goursat-based characterization in Section \ref{sec:Goursat} simplify significantly for subspace stabilizer codes?
    \item How does the Goursat data of a code behave under Clifford deformations?
    \item Is there a concise characterization of subsystem stabilizer codes with maximal (minimal) stabilizer, as introduced in Definition \ref{def:maxmincenter}?
    \item Do subsystem stabilizer codes with maximal (minimal) stabilizer have nice structural properties, analogous to those described in Proposition \ref{prop:CSSStructure} for subsystem CSS codes?
%\vva{Why cite HKS21? I couldn't find anything about stab codes over rings, just fields.}
    \item     Our doubling map holds for all modular-qudit subsystem stabilizer codes \eczoo{qudit_subsystem_stabilizer} --- can we generalize our other results to modular-qudit codes or operator algebra stabilizer codes \cite{KS08,DKV24} \eczoo{qubit_stabilizer_oaqecc}? For a generalization to modular-qudit codes, it is helpful that the theory of classical linear codes, subsystem stabilizer codes, and Goursat's Lemma all generalize to rings of the form $(\dsZ / m \dsZ)^n$ \cite{SAS17, Ellison23, MG22}, but many of our proofs in Sections \ref{sec:CSSRecovery} and \ref{sec:SubsystemStabilizerCode} rely on the vector space structure of the underlying Pauli group. In particular, the lack of a basis for modular-qudit Pauli groups may complicate the existing proofs in this work. For a generalization to operator algebra stabilizer codes, one may need to add to our perspective of subsystem stabilizer codes as subspaces of vector spaces (Section \ref{sec:PrelimSSCVS}), since this perspective ignores phases and thus is insensitive to the particular factor of the logical algebra used to encode information (see Eqs.~(\ref{eq:StabilizerSpaceDecomposition}) and (\ref{eq:StabilizerAlgebraDecomposition})). Since operator algebra codes encode information in more than one such factor in general, one must keep track of these factors when generalizing our results to operator algebra codes. One may also investigate how to transition from the Schrödinger picture, in which this paper is expressed, to the Heisenberg picture, in which the theory of operator algebra codes is usually expressed.
    \item Using existing (non-CSS) Floquet codes~\cite{Hastings2021,Zhang22,Bauer2023,Sullivan23,Ellison23_2,Davydova23,Dua23} \eczoo{floquet}, new CSS Floquet codes~\cite{Aasen2022,Davydova22,brown_2022,Bombin23,Townsend_23} can be constructed using our mapping. Such CSS codes could be more amenable to the analysis of fault-tolerance using ZX calculus~\cite{Bombin23,Bauer2023,Townsend_23}.   %Beyond subsystem stabilizer codes, the mapping codes could prove interesting for construction of new Floquet codes \eczoo{floquet} from existing ones. \NT{How many papers do we want to cite here?} \vva{Hasting/Haah and all of Nat's papers, of course :)}

\end{enumerate}

\begin{acknowledgments}
We thank Nathan Pflueger, Kenneth Brown, Andrew Nemec, and Eric Sabo for helpful feedback and suggestions.
MLL acknowledges William Gasarch and the REU-CAAR program at the University of Maryland, College Park.
NT is supported by the Walter Burke Institute for Theoretical Physics at Caltech.
Contributions to this work by NIST, an agency of the US government, are not subject to US copyright.
Any mention of commercial products does not indicate endorsement by NIST. 
VVA thanks Ryhor Kandratsenia and Olga Albert for providing daycare support throughout this work.
\end{acknowledgments}

\newpage~\newpage
\appendix

\section{Proofs of Propositions}\label{app:Proofs}

\subsection{Proof of Proposition \ref{prop:VSTowerParameters}}\label{app:ProofVSTowerParameters}
Write
\begin{align}\overline{H} &\coloneqq \left( H + H^{\omega} \right) / H \text{ and} \\ \underline{H} &\coloneqq H / \left(H \cap H^{\omega} \right).\end{align}
One verifies that the map
\begin{equation}\left(a+\left(H\cap H^{\om} \right), b+\left(H\cap H^{\om} \right)\right) \mapsto \om (a, b)\end{equation}
is an antisymmetric form on \(\underline{H}\). Since every vector space with an antisymmetric form has even dimension \cite{TA86}, we have \(\dim \underline{H}=2r\) for some \(r\in \dsZ\). Moreover, observe that
\begin{equation}
\begin{aligned}
    \dim \overline{H} + \dim \underline{H} &= \dim H + H^\om - \dim H \cap H^\om \\
    &= \dim G \times G - 2 \dim H \cap H^\om,
\end{aligned}
\end{equation}
so \(\dim \overline{H}=2k\) for some \(k\in \dsZ\). Next, note that
\begin{equation}\dim G \times G-\dim \left(H+H^{\om}\right) = \dim \left(H\cap H^{\om}\right).\end{equation}
Since the dimensions in Eq.~(\ref{eq:HTower}) must sum to \(\dim G \times G = 2n\), we find that \(\dim H\cap H^{\om} = n -k-r\), as needed. The final statement follows by defining $d$ appropriately.

\subsection{Proof of Proposition \ref{prop:CSSStructure}}\label{app:ProofCSSStructure}

First, note that Eq.~(\ref{eq:HTowerCSS}) follows directly from Eq.~(\ref{eq:DirectProductOmega}). Next, observe that
\begin{equation}
\begin{aligned}
    \dim {\parenth{H_X + H_Z^\theta}}/H_X &=
    \dim H_X^\theta / \left(H_Z\cap H_X^\theta\right) \\
    &= \dim {\parenth{H_Z + H_X^\theta}}/H_Z
\end{aligned}
\end{equation}
and
\begin{equation}
\begin{aligned}
    \dim H_X/{\parenth{H_X \cap H_Z^\theta}} &=
    \dim {\parenth{H_Z + H_X^\theta}}/H_X^\theta \\
    &= \dim H_Z/{\parenth{H_Z \cap H_X^\theta}},
\end{aligned}
\end{equation}
so the claim about dimensions follows. Finally, recall that
\begin{equation}d\coloneqq \min \swt\left({\parenth{H+H^{\om}}}\setminus H\right).\end{equation}
Fix \((a,b)\in {\parenth{H+H^{\om}}}\setminus H\) with \(\swt(a,b)=d\). Then \(a\in {\parenth{H_X+H_Z^{\theta}}}\) and \(b\in
{\parenth{H_Z+H_X^{\theta}}}\), but either \(a\notin H_X\) or \(b\notin H_Z\). If \(a\notin H_X\), then
\begin{equation}
 d = \swt(a,b) \geq \swt(a,0) = \wt(a) \geq d^{H_X}.
\end{equation}
If instead \(b\notin H_Z\), a similar argument shows that \(d \geq d^{H_Z}\). Conversely, suppose that \(a\in \left(H_X+H_Z^{\theta
}\right)\setminus H_X\) satisfies \(\wt(a)=\min \left\{ d^{H_X}, d^{H_Z} \right\}\). Then
\begin{equation}
d \leq \swt(a,0) = \wt(a) = \min \left\{ d^{H_X}, d^{H_Z} \right\}.
\end{equation}
If instead \(b\in {\parenth{H_Z+H_X^{\theta}}}\setminus H_Z\) satisfies \(\wt(b)=\min \left\{ d^{H_X}, d^{H_Z} \right\}\), a similar argument shows that \(d\leq \min \left\{ d^{H_X}, d^{H_Z} \right\}\).

\subsection{Proof of Proposition \ref{prop:ClassicalLinearCodeCorrectability}}\label{app:proofblah}

Since we know \(F a\), we can determine \(a+K\) by the first isomorphism theorem. By assumption, there exists $r \in R$ such that \(\wt(a+r) < \frac{d^R}{2}\). Thus, we can identify some $a+k \in a+K$ with $\wt(a+k) < \frac{d^R}{2}$, where $k\in K$ is unknown. By the triangle inequality, we have
\begin{equation}
\begin{aligned}
        \wt (r - k) &= \wt ((a+r) - (a+k)) \\&\leq \wt(a+r) + \wt(a+k) \\&< d^R \\&= \min \wt(K \setminus R).
\end{aligned}
\end{equation}
Thus, since $r - k \in K$, we must have $r - k \in R$, so $k \in R$.

\subsection{Proof of Proposition \ref{prop:GoursatProperties}} \label{app:ProofGoursatProperties}

On the one hand, observe that
\begin{equation}
\begin{aligned}
\phi \left(\sum_j c_j{e_X^j}+{n_X}+{N_X}\right) &= \sum_j c_j\phi \left({e_X^j}+{N_X}\right) \\
&= \sum_j c_j\left({e_Z^j}+{N_Z}\right) \\
&= \sum_j c_j{e_Z^j}+{n_Z}+{N_Z}.
\end{aligned}
\end{equation}
Conversely, suppose that \(\left({e_X},{e_Z}\right)\in {E_X}\times {E_Z}\) satisfies \(\phi \left({e_X}+{N_X}\right)={e_Z}+{N_Z}\). Write
\begin{equation}
    {e_X}+{N_X}=\sum_j c_j\left({e_X^j}+{N_X}\right)
\end{equation}
and
\begin{equation}
    {e_Z}+{N_Z}=\sum_j d_j\left({e_Z^j}+{N_Z}\right)
\end{equation}
for some \(c_j,d_j\in \dsF_p\). Then
\begin{equation}
    \sum_j c_j\left({e_Z^j}+{N_Z}\right)=\sum_j d_j\left({e_Z^j}+{N_Z}\right),
\end{equation}
so \(c_j=d_j\) for all \(j\). Then
\begin{equation}
    {e_X}=\sum_j c_j{e_X^j}+{n_X}
\end{equation}
and
\begin{equation}
    {e_Z}=\sum_j c_j{e_Z^j}+{n_Z}
\end{equation}
for some \({n_X}\in {N_X}\) and \({n_Z}\in {N_Z}\), so Eq.~(\ref{eq:GoursatHExplicit}) holds.

\subsection{Proof of Proposition \ref{prop:ExplicitGoursat}}\label{app:ProofExplicitGoursat}
One verifies Eq.~(\ref{eq:ExplicitGoursat}) from the definitions in Eq.~(\ref{eq:defsr}). To prove Eq.~(\ref{eq:KernelSum}), we need the following technical lemma.

\begin{lemma}
Let \(W\leq V\leq U\), and let \(\pi :U\to X\) be any linear map. Then \(\pi (W)=\pi (V)\) iff \((V\cap \ker \pi )+W=V\).
\end{lemma}
\begin{proof}
Suppose that \(\pi (W)=\pi (V)\). Fix \(v\in V\). Then there exists \(w\in W\) such that \(\pi (v)=\pi (w)\). Then \(v=(v-w)+w\), where \(v-w\in V\cap \ker \pi\). Conversely, suppose that \((V\cap \ker \pi )+W=V\). Fix \(v\in V\). Then there exist \(k\in V\cap \ker \pi\) and \(w\in W\) such that \(v=k+w\). Then \(\pi (v)=\pi (w)\). Thus, \(\pi (V)\subseteq \pi (W)\), as needed.
\end{proof}

Now, recall that $H$ is a direct product iff $\pi_X{\parenth{\ker \pi_Z}} = \pi_X{\parenth{\dsF_p^l}}$. Applying the above lemma with $U=V=\dsF_p^l$, $W = \ker \pi_Z$, and $\pi = \pi_X$, we obtain Eq.~(\ref{eq:KernelSum}).

\subsection{Proof of Proposition \ref{prop:GoursatComplement}} \label{app:ProofGoursatComplement}
Fix a basis \(\left\{{e_X^j}+{N_X}\right\}_j\) for \({E_X}/{N_X}\). Define the function
\begin{equation}
\begin{aligned}
\tilde{\phi}_X :{N_X^\theta} &\to {E_X}/{N_X} \\
{n_X^\theta} &\mapsto \sum_j \theta \left({e_X^j},{n_X^\theta} \right)\left({e_X^j}+{N_X}\right).
\end{aligned}
\end{equation}
Since \({n_X^\theta} \in {N_X^\theta} \), the function \(\tilde{\phi}_X\) is independent of the choice of representatives \({e_X^j}\) for the cosets
\({e_X^j}+{N_X}\). Since \(\theta\) is bilinear, \(\tilde{\phi}_X\) is linear. Now, if \({n_X^\theta} \in {E_X^\theta} \), then \(\tilde{\phi}_X\left({n_X^\theta}\right)={N_X}\). Conversely, suppose that \(\tilde{\phi}_X\left({n_X^\theta} \right)={N_X}\). Since the \({e_X^j}+{N_X}\) are linearly independent, we have \(\theta \left({e_X^j},{n_X^\theta} \right)=0\) for all \(j\). Since every \({e_X}\in {E_X}\) can be written as
\begin{equation}{e_X}=\sum_j c_j{e_X^j}+{n_X}\end{equation}
for some \(c_j\in \dsF_p\) and \({n_X}\in {N_X}\), we find that \({n_X^\theta} \in {E_X^\theta} \). Thus, \(\ker \tilde{\phi}_X={E_X^\theta} \).
Then \(\tilde{\phi}_X\) induces the injection
\begin{equation}\label{eq:PhiXDef}
\begin{aligned}
\phi_X : {N_X^\theta} / {E_X^\theta} &\to {E_X} / {N_X} \\
{n_X^\theta} + {E_X^\theta} &\mapsto \tilde{\phi}_X\left({n_X^\theta}\right).
\end{aligned}
\end{equation}
In fact, since \(\dim {N_X^\theta} /{E_X^\theta} =\dim {E_X}/{N_X}\), the map \(\phi_X\) is an isomorphism.

Similarly, given the basis
\begin{equation}
    \left\{\phi
\left({e_X^j}+{N_X}\right) \eqqcolon {e_Z^j}+{N_Z}\right\}_j
\end{equation}
for \({E_Z}/{N_Z}\), there exists an isomorphism \(\phi_Z:{N_Z^\theta} /{E_Z^\theta} \to {E_Z}/{N_Z}\), where
\begin{equation}\label{eq:PhiZDef}
\phi_Z\left({n_Z^\theta} +{E_Z^\theta} \right)\coloneqq \sum_j \theta \left({e_Z^j},{n_Z^\theta} \right)\left({e_Z^j}+{N_Z}\right).
\end{equation}

Now, define \(\phi ^{\om} \coloneqq \phi_X^{-1}\phi ^{-1}\phi_Z\). Fix \(\left({n_Z^\theta} ,{n_X^\theta} \right)\in {N_Z^\theta}\times {N_X^\theta} \). By definition, we have
\begin{equation}\phi ^{\om} \left({n_Z^\theta} +{E_Z^\theta} \right) = {n_X^\theta} +{E_X^\theta}\end{equation}
iff
\begin{equation}\phi_Z\left({n_Z^\theta} +{E_Z^\theta} \right) = \phi \left( \phi_X\left({n_X^\theta} +{E_X^\theta} \right)\right).
\end{equation}
By Eqs.~(\ref{eq:PhiXDef}) and (\ref{eq:PhiZDef}), this is equivalent to
\begin{equation}\begin{aligned}
    \sum_j \theta \left({e_Z^j},{n_Z^\theta} \right)&\left({e_Z^j}+{N_Z}\right) \\ = {}&\phi \left(\sum_j \theta \left({e_X^j},{n_X^\theta} \right)\left({e_X^j}+{N_X}\right)\right),
\end{aligned}\end{equation}
which is equivalent to
\begin{equation}\begin{aligned}
    \sum_j \theta \left({e_Z^j},{n_Z^\theta} \right){}&\phi \left({e_X^j}+{N_X}\right) \\ = \sum_j {}&\theta \left({e_X^j},{n_X^\theta} \right)\phi \left({e_X^j}+{N_X}\right).
\end{aligned}\end{equation}
Since the $\phi \left({e_X^j}+{N_X}\right)$ are linearly independent, this is equivalent to
\begin{equation}\forall j, \,  \theta \left({e_Z^j},{n_Z^\theta} \right) = \theta \left({e_X^j},{n_X^\theta} \right),\end{equation}
which by Eq.~(\ref{eq:GoursatHExplicit}) is equivalent to
\begin{equation}\forall \left({e_X},{e_Z}\right)\in H, \, \theta \left({e_Z},{n_Z^\theta} \right) = \theta \left({e_X},{n_X^\theta}\right).\end{equation}
Rewriting, this occurs iff
\begin{equation}\forall \left({e_X},{e_Z}\right)\in H, \, \om \left(\left({e_X},{e_Z}\right),\left({n_Z^\theta} ,{n_X^\theta}\right)\right) = 0,\end{equation}
which is equivalent to $\left({n_Z^\theta} ,{n_X^\theta} \right) \in H^{\om}$. Now, since \({N_X}\times {N_Z}\leq H\), we have \(H^{\om} \leq {N_Z^\theta} \times {N_X^\theta} \). Thus, we find
\begin{equation}H^{\om} = \left\{\left({n_Z^\theta} ,{n_X^\theta} \right)\in {N_Z^\theta} \times {N_X^\theta} \, \middle| \, \phi ^{\om} \left({n_Z^\theta} +{E_Z}^{\theta
}\right)={n_X^\theta} +{E_X^\theta} \right\},\end{equation}
as needed.

\subsection{Proof of Proposition \ref{prop:GoursatIntersection}} \label{app:ProofGoursatIntersection}
Consider the injections \(\phi_T,\phi_W\) defined in the diagram
\begin{widetext}
\begin{equation}
\begin{array}{ccccccc}
 \left({E_X}\cap \tilde{{E_X}}\right)/\left({N_X}\cap \tilde{{N_X}}\right) & \overset{\phi_T}{\hookrightarrow} & \left({E_X}/{N_X}\right)\times \left(\tilde{{E_X}}/\tilde{{N_X}}\right)
& \overset{\phi \times \tilde{\phi}}{\xrightarrow{\sim}} & \left({E_Z}/{N_Z}\right)\times \left(\tilde{{E_Z}}/\tilde{{N_Z}}\right) & \overset{\phi_W}{\hookleftarrow}
& \left({E_Z}\cap \tilde{{E_Z}}\right)/\left({N_Z}\cap \tilde{{N_Z}}\right) \\
 a+\left({N_X}\cap \tilde{{N_X}}\right) & {\mapsto} & \left(a+{N_X},a+\tilde{{N_X}}\right) & & \left(b+{N_Z},b+\tilde{{N_Z}}\right)
& {\mapsfrom} & b+\left({N_Z}\cap \tilde{{N_Z}}\right)
\end{array}.
\end{equation}
\end{widetext}
Moving ``left to right" in the above diagram, define
\begin{equation}\scW\coloneqq \phi_W^{-1}\left(\phi \times \tilde{\phi} \right)\phi_T\left(\left({E_X}\cap \tilde{{E_X}}\right)/\left({N_X}\cap \tilde{{N_X}}\right)\right).\end{equation}
Moving ``right to left" in the above diagram, define
\begin{align}
\scV &\coloneqq \phi_W(\scW), \\
\scU &\coloneqq \left(\phi \times \tilde{\phi} \right)^{-1}(\scV), \\
\scT &\coloneqq \phi_T^{-1}(\scU).
\end{align}
From the definitions, it is straightforward to verify that the appropriate restrictions of \(\phi_T,\phi \times \tilde{\phi} ,\phi_W\) satisfy
\begin{equation}\begin{array}{ccccccc}
 \scT & \overset{\phi_T}{\xrightarrow{\sim}} & \scU & \overset{\phi \times \tilde{\phi}}{\xrightarrow{\sim}} & \scV & \overset{\phi
_W}{\xleftarrow{\sim}} & \scW \\
\end{array}
.\end{equation}
By the lattice isomorphism theorem, there exist subspaces
\begin{equation}\begin{array}{ccc}
 {E_X}\cap \tilde{{E_X}} & & {E_Z}\cap \tilde{{E_Z}} \\
 | & & | \\
 T & \text{and} & W \\
 | & & | \\
 {N_X}\cap \tilde{{N_X}} & & {N_Z}\cap \tilde{{N_Z}} \\
\end{array}\end{equation}
such that
\begin{equation}
\begin{aligned}
\scT&=T\left/\left({N_X}\cap \tilde{{N_X}}\right)\right. \text{and} \\ \scW&=W\left/\left({N_Z}\cap \tilde{{N_Z}}\right).\right.
\end{aligned}
\end{equation}
Now, fix \((a,b)\in T\times W\). Then with
\begin{equation}
    \phi ^{\cap} \coloneqq \phi_W^{-1}\left(\phi \times \tilde{\phi} \right)\phi_T,
\end{equation}
we have
\begin{equation}\phi^{\cap} \left(a+\left({N_X}\cap \tilde{{N_X}}\right)\right) = b+\left({N_Z}\cap \tilde{{N_Z}}\right)\end{equation}
iff
\begin{equation}\left(\phi \times \tilde{\phi}\right)\phi_T\left(a+\left({N_X}\cap \tilde{{N_X}}\right)\right) = \phi_W\left(b+\left({N_Z}\cap \tilde{{N_Z}}\right)\right),\end{equation}
which occurs iff
\begin{equation}\left(\phi \times \tilde{\phi} \right)\left(a+{N_X},a+\tilde{{N_X}}\right) = \left(b+{N_Z},b+\tilde{{N_Z}}\right).\end{equation}
Reading coordinate-wise, this occurs iff
\begin{equation}
\begin{aligned}
    \phi \left(a+{N_X}\right) &= b+{N_Z} \text{ and} \\ \tilde{\phi} \left(a+\tilde{{N_X}}\right) &= b+\tilde{{N_Z}},
\end{aligned}
\end{equation}
which occurs iff
\begin{equation}(a,b)\in H \text{ and } (a,b)\in \tilde{H}.\end{equation}
Moreover, if \((a,b)\in H\cap \tilde{H}\subseteq {E_X}\cap \tilde{{E_X}}\), then by the same computation above, we find
\begin{equation}\begin{aligned}
\left(\phi \times \tilde{\phi}\right)\left(\phi_T\left(a+\left({N_X}\cap \tilde{{N_X}}\right)\right)\right)& \\ = \phi_W\left(b+\left({N_Z}\cap \tilde{{N_Z}}\right)\right)&.
\end{aligned}\end{equation}
That is, \(b+\left({N_Z}\cap
\tilde{{N_Z}}\right)\in \scW\), so \(a+\left({N_X}\cap \tilde{{N_X}}\right)\in \scT\). Then \((a,b)\in T\times W\), so \(H\cap \tilde{H}\subseteq T\times W\). Thus, we have
\begin{equation}
\begin{aligned}
H\cap \tilde{H}=\Big\{(a,b)\in T\times W \, \Big| \, \phi ^{\cap} \Big(&a+\Big({N_X}\cap \tilde{{N_X}}\Big)\Big) \\ = {} &b+\Big({N_Z}\cap \tilde{{N_Z}}\Big)\Big\},
\end{aligned}
\end{equation}
as needed.

\subsection{Proof of Proposition \ref{prop:minimalmaximal}}\label{app:proofminimalmaximal}

Suppose that \(H\) has both maximal and minimal stabilizer. Taking orthogonal complements in the right-hand side factor, we have
\begin{align}
{N_X}\cap {E_Z^\theta} &= {E_X}\cap {N_Z^\theta} \text{ and} \\
{E_X} + {N_Z^\theta} &= {N_X} + {E_Z^\theta}.
\end{align}
On the one hand,
\begin{equation}
\begin{aligned}
\dim {E_X}+\dim {N_Z^\theta} &= \dim \left({E_X}+{N_Z^\theta} \right) + \dim \left({E_X}\cap {N_Z^\theta} \right) \\
&= \dim \left({N_X}+{E_Z^\theta} \right) + \dim \left({N_X}\cap {E_Z^\theta} \right) \\
&= \dim {N_X}+\dim {E_Z^\theta} .
\end{aligned}
\end{equation}
On the other hand,
\begin{align}
\dim {E_X} &\geq \dim {N_X} \text{ and} \\
\dim {N_Z^\theta} &\geq \dim {E_Z^\theta}.
\end{align}
Thus, we have
\begin{equation}\label{eq:farm}
{E_X} = {N_X} \text{ and } {E_Z} = {N_Z},
\end{equation}
as needed. The converse is clear --- assume Eq.~(\ref{eq:farm}) and simplify the right-hand side of Eq.~(\ref{eq:macdonal}).

\section{Generalizing Theorem \ref{thm:mapping} to subsystem stabilizer codes over modular qudits}\label{app:Generalize}

In this appendix, we briefly describe how to generalize Theorem \ref{thm:mapping} to subsystem stabilizer codes over modular qudits. The material in the preliminaries, Section \ref{sec:Prelim}, generalizes as follows \cite{DKV24,TA86,Ellison23}.
\begin{itemize}
    \item The prime $p$ is now an arbitrary integer $m$, so $G$ is now the finite abelian group $(\dsZ / m \dsZ)^n$ with standard generating set $\{(1,0,\ldots,0),\ldots,(0,\ldots,0,1)\}$
    \item The Pauli group modulo phases is still isomorphic to $G \times G$
    \item Eq.~(\ref{eq:xicomplementdim}) now reads as $\abs{H} \abs{H^\xi} = \abs{G}$, where the absolute value denotes group order
    \item A subsystem stabilizer code is now a subgroup $H \leq G \times G$. We call $H$ a subsystem CSS code if $H$ is a direct product of subgroups of $G$
    \item We no longer keep track of the number of logical and gauge qudits in a subsystem stabilizer code. Rather, we keep track of the \emph{logical Pauli group} $\overline{H} \coloneqq \left( H + H^{\omega} \right) / H$ and the \emph{gauge Pauli group} $\underline{H} \coloneqq H / \left(H \cap H^{\omega} \right)$ --- see Eqs.~(\ref{eq:logicalpauli}) and (\ref{eq:gaugepauli}). The logical (resp. gauge) Pauli group is isomorphic to the group of Pauli operators on the logical (resp. gauge) subsystem $\dsL$ (resp. $\dsG$) \cite{Ellison23}
\end{itemize}

Having established the necessary background, we now sketch the generalization of Theorem \ref{thm:mapping}.
\begin{itemize}
    \item Eq.~(\ref{eq:OmegaDimension}) in Lemma \ref{lem:DeltaLemma} now reads as $\abs{\Delta(H)}=\abs{H}^2$, and the analogue of Lemma \ref{lem:DeltaLemma} holds with essentially the same proof
    \item Theorem \ref{thm:mapping} now reads as follows, where the proof is similar
\end{itemize}
\begin{theorem}
Let \(H\leq G\times G\) be a subsystem stabilizer code with logical Pauli group $\overline{H}$, gauge Pauli group $\underline{H}$, and distance $d$. Then $\Delta(H)$ is a subsystem CSS code with logical Pauli group isomorphic to $\overline{H} \times \overline{H}$, gauge Pauli group isomorphic to $\underline{H} \times \underline{H}$, and distance $d'$, where $d \leq d' \leq 2d$. Moreover, if \(H\) admits a collection of generators with symplectic weight at most $w$, then \(\Delta (H)\) admits a collection of generators with symplectic weight at most $2w$.
\end{theorem}

\section{An alternative classical step for error recovery}\label{app:Par}

In this appendix, we explore an alternative to the classical step for error recovery described in Section \ref{sec:CSSRecoveryClassical}. To begin, consider the commutative diagram of canonical quotients
\begin{equation}
    \begin{tikzcd}
    G \arrow[dd,"\overline{*}"] \arrow[dr, "\Syn_X"] \\
    & G/{\parenth{H_X + H_Z^\theta}}\\
    G/H_X \arrow[ur, swap,"\Par_X"]
\end{tikzcd}.
\end{equation}
In Section \ref{sec:CSSRecoveryClassical}, we used the parity check $\Syn_X$ of the classical linear code $H_X + H_Z^\theta$ to recover an error class $a+H_X$ from its syndrome $a+\parenth{H_X+H_Z^{\theta}}$. Here, we investigate the utility of $\Par_X$ for the same task --- this line of inquiry is motivated by the Bacon--Shor decoder \cite{B06}, which utilizes $\Par_X$ for error recovery (see Example \ref{ex:baconshor}). We will need the following notation.

Let $\sigma \subseteq G$ denote the standard basis for $G$. Let $\sigma_0$ be any subset of $\sigma$ such that $\overline{\sigma_0}$ is a basis for $G/H_X$, where the over-line denotes quotient by $H_X$. For any $a+H_X \in G/H_X$, we write $\wt_{\overline{\sigma_0}}(a+H_X)$ to denote the weight of $a+H_X$ in the basis $\overline{\sigma_0}$, i.e., the number of nonzero coordinates of $a+H_X$ in the basis $\overline{\sigma_0}$.

Now, recall that by Proposition \ref{prop:ClassicalLinearCodeCorrectability}, $\Syn_X$ can recover the error class $a+H_X$ from its syndrome $a+\parenth{H_X+H_Z^{\theta}}$ provided that
\begin{equation}
    \min \wt(a+H_X) < \frac{d^{H_X}}{2}. \label{eq:SynSuff}
\end{equation}
Similarly, $\Par_X$ can recover the error \(a+H_X\) from its syndrome \(a+\parenth{H_X+H_Z^{\theta}}\) provided that
\begin{equation}
    \wt_{\overline{\sigma_0}}\left(a+H_X\right) < \frac{d^{\Par_X}}{2}, \label{eq:ParSuff}
\end{equation}
where
\begin{equation}
d^{\Par_X}\coloneqq \min \wt_{\overline{\sigma_0}}\left({\parenth{H_X+H_Z^{\theta}}}/H_X\right)
\end{equation}
is the distance of $\Par_X$ in the basis $\overline{\sigma_0}$. Our first observation is that the notions of weight and distance in Eqs.~(\ref{eq:SynSuff}) and (\ref{eq:ParSuff}) are related.

\begin{proposition}\label{prop:CompareSynPar}
We have
\begin{equation}
    \min \wt(a+H_X) \leq \wt_{\overline{\sigma_0}}\left(a+H_X\right)
\end{equation}
for all $a \in G$. Moreover, we have
\begin{equation}
        d^{H_X} \leq d^{\Par_X}.
\end{equation}
\end{proposition}
\begin{proof}
First, note that for any $a \in G$, we have
\begin{equation}
    a+H_X=\sum_{\alpha \in {\sigma_0}} a_{\alpha} \left(\alpha+H_X\right)
\end{equation}
for some $a_\alpha \in \dsF_p$, so
\begin{equation}
    a + h_X =\sum_{\alpha \in {\sigma_0}} a_{\alpha} \alpha
\end{equation}
for some $h_X \in H_X$. Thus,
\begin{equation}
    \min \wt(a+H_X) \leq \wt(a+h_X) = \wt_{\overline{\sigma_0}}\left(a+H_X\right).
\end{equation}
Now, suppose in addition that 
\begin{equation}
a+H_X\in {\parenth{H_X+H_Z^{\theta}}}/H_X
\end{equation}
is nonzero and satisfies
\begin{equation}
    \wt_{\overline{\sigma_0}}\left(a+H_X\right) = d^{\Par_X}.
\end{equation}
Then $a+h_X\in {\parenth{H_X+H_Z^{\theta}}}\setminus H_X$, so
\begin{equation}
    d^{H_X} \leq \wt(a+h_X) = \wt_{\overline{\sigma_0}}\left(a+H_X\right) = d^{\Par_X}.
\end{equation}
\end{proof}

Unfortunately, this proposition does not immediately reveal whether $\Par_X$ performs better or worse than $\Syn_X$. Indeed, the distance of $\Par_X$ is nominally lower bounded by the distance of $\Syn_X$, but the notion of error weight relevant for $\Par_X$ is also lower bounded by that for $\Syn_X$. Thus, it is not immediately clear whether one code performs better than the other. To shed light on this issue, we investigate the special case in which the notions of error weight for $\Par_X$ and $\Syn_X$ coincide.

\begin{definition}\label{def:GoodBasisDefinition}
We say that $H_X \leq G$ \emph{respects weight} if there exists $\sigma_0 \subseteq \sigma$ such that \(\overline{\sigma_0}\) is a basis for $G/H_X$ satisfying
\begin{equation}\min \wt(a+H_X) = \wt_{\overline{\sigma_0}}\left(a+H_X\right)\end{equation}
for all \(a\in G\).
\end{definition}

Note that by definition, $H_X$ respects weight iff the notions of weight for $\Syn_X$ and $\Par_X$ in Eqs.~(\ref{eq:SynSuff}) and (\ref{eq:ParSuff}) coincide. In fact, if $H_X$ respects weight, then the distances of $\Syn_X$ and $\Par_X$ also coincide. To see this, note that
\begin{equation}
\begin{aligned}
d^{\Par_X}&= \underset{a\in (H_X+H_Z^{\theta})\setminus H_X}{\min} \wt_{\overline{\sigma_0}}\left(a+H_X\right) \\
&\leq \underset{a\in (H_X+H_Z^{\theta})\setminus H_X}{\min} \wt(a) \\
&= d^{H_X},
\end{aligned}
\end{equation}
so by Proposition \ref{prop:CompareSynPar}, we have $d^{H_X} = d^{\Par_X}$.

Interestingly, it turns out that there are not very many subspaces $H_X \leq G$ that respect weight.

\begin{proposition}
\label{prop:goodcharacterization}
A subspace $H_X \leq G$ respects weight iff $H_X$ admits a basis whose elements have weight at most $2$.
\end{proposition}
\begin{proof}
$(\implies)$ Fix $\beta \in \sigma_0^c \coloneqq \sigma \setminus \sigma_0$ and write
\begin{equation}
    \beta + H_X = \sum_{\alpha \in {\sigma_0}} a_{\alpha \beta} \left(\alpha + H_X\right),
\end{equation}
where the $a_{\alpha \beta} \in \dsF_p$. If more than one $a_{\alpha \beta}$ is nonzero, then $\wt_{\overline{\sigma_0}}\left(\beta+H_X\right)>1$ but $\wt(\beta) = 1$, a contradiction. Thus, at most one $a_{\alpha \beta}$ is nonzero, so $\beta = b_{\beta} \alpha_\beta + \gamma_\beta$ for some $b_{\beta} \in \dsF_p$, $\alpha_\beta\in \sigma_0$ and $\gamma_\beta \in H_X$. We claim that $\Gamma \coloneqq \left\{\gamma_\beta \, \middle| \, \beta \in \sigma_0^c\right\} \subseteq H_X$ is linearly independent. Indeed, suppose that
\begin{equation}
    \sum_{\beta \in \sigma_0^c} c_\beta \gamma_\beta = 0
\end{equation}
for some $c_\beta \in \dsF_p$. Then
\begin{equation}\label{eq:blahblah}
    \sum_{\beta \in \sigma_0^c} c_\beta \beta = \sum_{\beta \in \sigma_0^c} c_\beta b_{\beta} \alpha_\beta,
\end{equation}
which by linear independence of $\sigma$ implies that $c_\beta=0$ for all $\beta \in \sigma_0^c$. Thus, $\Gamma$ has exactly $\abs{\sigma_0^c} = \dim H_X$ elements, so $\Gamma$ is a basis for $H_X$ whose elements have weight at most $2$.

$(\impliedby)$ Let $M$ be the $\dim H_X \times \dim G$ matrix whose rows are the basis vectors of $H_X$. Let $M'$ be the reduced row echelon form of $M$. Note that the row space of $M'$ is $H_X$. Let ${\sigma_0}$ be the subset of $\sigma$ corresponding to the non-pivot columns of $M'$, so $\sigma_0^c$ corresponds to the pivot columns of $M'$. Since each row of $M$ has weight at most $2$, it follows that for each pivot $\beta \in \sigma_0^c$, there is at most one other nonzero entry (which lies to the right of $\beta$ and is not a pivot) in the corresponding row of $M'$. Call that entry $-b_{\beta} \alpha_\beta$, where $b_{\beta} \in \dsF_p$ and $\alpha_\beta\in \sigma_0$. Now, fix any $a \in G$, and write
\begin{equation}
a = \sum_{\alpha \in \sigma_0
} a_\alpha \alpha + \sum_{\beta \in \sigma_0^c
} a_\beta \beta,
\end{equation}
where the $a_\alpha, a_\beta \in \dsF_p$. Then
\begin{equation}
\begin{aligned}
    a + H_X &= \sum_{\alpha \in \sigma_0
} a_\alpha \left( \alpha + H_X\right) + \sum_{\beta \in \sigma_0^c
} a_\beta \left( \beta + H_X\right)\\
&=\sum_{\alpha \in \sigma_0
} a_\alpha \left( \alpha + H_X\right) + \sum_{\beta \in \sigma_0^c
} a_\beta \left(b_{\beta} \alpha_\beta + H_X\right).
\end{aligned}
\end{equation}
Thus, the weight of $a+H_X$ is at most the number of nonzero $a_\alpha$ plus the number of nonzero $a_\beta$. That is,
\begin{equation}\wt_{\overline{\sigma_0}}\left(a+H_X\right)\leq \wt(a).\end{equation}
Since this holds for any $a \in G$, we conclude that $H_X$ respects weight.
\end{proof}

In the following example, we illustrate the usage of $\Par_X$ to correct errors affecting a subsystem CSS code whose underlying classical linear codes respect weight.

\begin{example}[Bacon--Shor code]\label{ex:baconshor}
The physical qubits of a Bacon--Shor code \cite{B06} are located at the vertices of an \(l\times l\) grid, so we set \(G = \dsF_2^{l^2}\). For concreteness, let $l=4$. Every pair of row-adjacent sites \((i,j),(i,j+1)\) on the grid contributes an \(X\)-type gauge generator \(X^{(i,j)}\otimes X^{(i,j+1)}\), and every pair of column-adjacent sites \((i,j),(i+1,j)\) on the grid contributes a \(Z\)-type gauge generator \(Z^{(i,j)}\otimes Z^{(i+1,j)}\). Thus, the gauge group is the direct product of
\begin{equation}
\begin{aligned}
&H_X \coloneqq \\ \Bigg\langle
&\left(\begin{array}{cccc}
    1 & 1 & 0 & 0 \\
    0 & 0 & 0 & 0 \\
    0 & 0 & 0 & 0 \\
    0 & 0 & 0 & 0 \\
\end{array}\right),
\left(\begin{array}{cccc}
    0 & 1 & 1 & 0 \\
    0 & 0 & 0 & 0 \\
    0 & 0 & 0 & 0 \\
    0 & 0 & 0 & 0 \\
\end{array}\right),
\left(\begin{array}{cccc}
    0 & 0 & 1 & 1 \\
    0 & 0 & 0 & 0 \\
    0 & 0 & 0 & 0 \\
    0 & 0 & 0 & 0 \\
\end{array}\right), \\ &\ldots, \\
&\left(\begin{array}{cccc}
    0 & 0 & 0 & 0 \\
    0 & 0 & 0 & 0 \\
    0 & 0 & 0 & 0 \\
    1 & 1 & 0 & 0 \\
\end{array}\right),
\left(\begin{array}{cccc}
    0 & 0 & 0 & 0 \\
    0 & 0 & 0 & 0 \\
    0 & 0 & 0 & 0 \\
    0 & 1 & 1 & 0 \\
\end{array}\right),
\left(\begin{array}{cccc}
    0 & 0 & 0 & 0 \\
    0 & 0 & 0 & 0 \\
    0 & 0 & 0 & 0 \\
    0 & 0 & 1 & 1 \\
\end{array}\right)
\Bigg\rangle
\end{aligned}
\end{equation}
and
\begin{equation}\label{eq:BaconShorCZ}
\begin{aligned}
&H_Z \coloneqq \\ \Bigg\langle
&\left(\begin{array}{cccc}
    1 & 0 & 0 & 0 \\
    1 & 0 & 0 & 0 \\
    0 & 0 & 0 & 0 \\
    0 & 0 & 0 & 0 \\
\end{array}\right),
\left(\begin{array}{cccc}
    0 & 0 & 0 & 0 \\
    1 & 0 & 0 & 0 \\
    1 & 0 & 0 & 0 \\
    0 & 0 & 0 & 0 \\
\end{array}\right),
\left(\begin{array}{cccc}
    0 & 0 & 0 & 0 \\
    0 & 0 & 0 & 0 \\
    1 & 0 & 0 & 0 \\
    1 & 0 & 0 & 0 \\
\end{array}\right), \\ &\ldots, \\
&\left(\begin{array}{cccc}
    0 & 0 & 0 & 1 \\
    0 & 0 & 0 & 1 \\
    0 & 0 & 0 & 0 \\
    0 & 0 & 0 & 0 \\
\end{array}\right),
\left(\begin{array}{cccc}
    0 & 0 & 0 & 0 \\
    0 & 0 & 0 & 1 \\
    0 & 0 & 0 & 1 \\
    0 & 0 & 0 & 0 \\
\end{array}\right),
\left(\begin{array}{cccc}
    0 & 0 & 0 & 0 \\
    0 & 0 & 0 & 0 \\
    0 & 0 & 0 & 1 \\
    0 & 0 & 0 & 1 \\
\end{array}\right)
\Bigg\rangle.
\end{aligned}
\end{equation}
Note that $H_X$ and $H_Z$ respect weight by Proposition \ref{prop:goodcharacterization}. Now, the complement \(H_X^{\theta} \) corresponds to the collection of all \(Z\)-type Pauli operators which commute with all the \(X\)-type Pauli operators in \(H_X\) (and similarly for \(H_Z^{\theta} \)). Explicitly, we have
\begin{equation}H_X^{\theta} =\Bigg\langle \left(
\begin{array}{cccc}
    1 & 1 & 1 & 1 \\
    0 & 0 & 0 & 0 \\
    0 & 0 & 0 & 0 \\
    0 & 0 & 0 & 0 \\
\end{array}
\right),\ldots ,\left(
\begin{array}{cccc}
    0 & 0 & 0 & 0 \\
    0 & 0 & 0 & 0 \\
    0 & 0 & 0 & 0 \\
    1 & 1 & 1 & 1 \\
\end{array}
\right)\Bigg\rangle\end{equation}
and
\begin{equation}H_Z^{\theta} =\Bigg\langle \left(
\begin{array}{cccc}
    1 & 0 & 0 & 0 \\
    1 & 0 & 0 & 0 \\
    1 & 0 & 0 & 0 \\
    1 & 0 & 0 & 0 \\
\end{array}
\right),\ldots ,\left(
\begin{array}{cccc}
    0 & 0 & 0 & 1 \\
    0 & 0 & 0 & 1 \\
    0 & 0 & 0 & 1 \\
    0 & 0 & 0 & 1 \\
\end{array}
\right)\Bigg\rangle .\end{equation}
Thus, we find that
\begin{equation}
H_X+H_Z^{\theta} = \Bigg\langle H_X,
\left(\begin{array}{cccc}
    1 & 0 & 0 & 0 \\
    1 & 0 & 0 & 0 \\
    1 & 0 & 0 & 0 \\
    1 & 0 & 0 & 0 \\
\end{array}\right)\Bigg\rangle,
\end{equation}
\begin{equation}
H_Z+H_X^{\theta} = \Bigg\langle H_Z,\left(
\begin{array}{cccc}
    1 & 1 & 1 & 1 \\
    0 & 0 & 0 & 0 \\
    0 & 0 & 0 & 0 \\
    0 & 0 & 0 & 0 \\
\end{array}\right)\Bigg\rangle,
\end{equation}
\begin{equation}
\begin{aligned}
&H_X \cap H_Z^{\theta} = \\ & \Bigg\langle \left(
\begin{array}{cccc}
    1 & 1 & 0 & 0 \\
    1 & 1 & 0 & 0 \\
    1 & 1 & 0 & 0 \\
    1 & 1 & 0 & 0 \\
\end{array}\right),
\left(\begin{array}{cccc}
    0 & 1 & 1 & 0 \\
    0 & 1 & 1 & 0 \\
    0 & 1 & 1 & 0 \\
    0 & 1 & 1 & 0 \\
\end{array}\right),
\left(\begin{array}{cccc}
    0 & 0 & 1 & 1 \\
    0 & 0 & 1 & 1 \\
    0 & 0 & 1 & 1 \\
    0 & 0 & 1 & 1 \\
\end{array}\right)
\Bigg\rangle,
\end{aligned}\end{equation}
and
\begin{equation}
\begin{aligned}
&H_Z \cap H_X^{\theta} = \\&\Bigg\langle \left(
\begin{array}{cccc}
    1 & 1 & 1 & 1 \\
    1 & 1 & 1 & 1 \\
    0 & 0 & 0 & 0 \\
    0 & 0 & 0 & 0 \\
\end{array}\right),
\left(\begin{array}{cccc}
    0 & 0 & 0 & 0 \\
    1 & 1 & 1 & 1 \\
    1 & 1 & 1 & 1 \\
    0 & 0 & 0 & 0 \\
\end{array}\right),
\left(\begin{array}{cccc}
    0 & 0 & 0 & 0 \\
    0 & 0 & 0 & 0 \\
    1 & 1 & 1 & 1 \\
    1 & 1 & 1 & 1 \\
\end{array}\right)
\Bigg\rangle.
\end{aligned}\end{equation}
Moreover, we have
\begin{equation}\label{eq:BaconShorBasis}
\begin{aligned}
G/H_X=\Bigg\langle
&\left(\begin{array}{cccc}
    1 & 0 & 0 & 0 \\
    0 & 0 & 0 & 0 \\
    0 & 0 & 0 & 0 \\
    0 & 0 & 0 & 0 \\
\end{array}
\right)+H_X,
\left(\begin{array}{cccc}
    0 & 0 & 0 & 0 \\
    1 & 0 & 0 & 0 \\
    0 & 0 & 0 & 0 \\
    0 & 0 & 0 & 0 \\
\end{array}
\right)+H_X,\\&
\left(\begin{array}{cccc}
    0 & 0 & 0 & 0 \\
    0 & 0 & 0 & 0 \\
    1 & 0 & 0 & 0 \\
    0 & 0 & 0 & 0 \\
\end{array}
\right)+H_X,
\left(\begin{array}{cccc}
    0 & 0 & 0 & 0 \\
    0 & 0 & 0 & 0 \\
    0 & 0 & 0 & 0 \\
    1 & 0 & 0 & 0 \\
\end{array}
\right)+H_X
\Bigg\rangle
\end{aligned}
\end{equation}
and
\begin{equation}\label{eq:BaconShorBasis2}
\begin{aligned}
G/{\parenth{H_X+H_Z^{\theta}}} = \Bigg\langle &\left(\begin{array}{cccc}
    1 & 0 & 0 & 0 \\
    0 & 0 & 0 & 0 \\
    0 & 0 & 0 & 0 \\
    0 & 0 & 0 & 0 \\
\end{array}
\right)+\parenth{H_X+H_Z^\theta},\\
&\left(\begin{array}{cccc}
    0 & 0 & 0 & 0 \\
    1 & 0 & 0 & 0 \\
    0 & 0 & 0 & 0 \\
    0 & 0 & 0 & 0 \\
\end{array}
\right)+\parenth{H_X+H_Z^\theta},\\
&\left(\begin{array}{cccc}
    0 & 0 & 0 & 0 \\
    0 & 0 & 0 & 0 \\
    1 & 0 & 0 & 0 \\
    0 & 0 & 0 & 0 \\
\end{array}
\right)+\parenth{H_X+H_Z^\theta}
\Bigg\rangle .
\end{aligned}
\end{equation}

Now, in the bases of Eqs.~(\ref{eq:BaconShorBasis}) and (\ref{eq:BaconShorBasis2}), the parity check \(\Par_X\) is represented by the ${(l-1) \times l}$ matrix
\begin{equation}\label{eq:ParXMatrixBS}
\Par_X=
\left(\begin{array}{cccc}
    1 & 0 & 0 & 1 \\
    0 & 1 & 0 & 1 \\
    0 & 0 & 1 & 1 \\
\end{array}\right).
\end{equation}
Abstractly, the kernel of \(\Par_X\) is
\begin{equation}{\parenth{H_X+H_Z^{\theta}}}/H_X=\Bigg\langle \left(
\begin{array}{cccc}
    1 & 0 & 0 & 0 \\
    1 & 0 & 0 & 0 \\
    1 & 0 & 0 & 0 \\
    1 & 0 & 0 & 0 \\
\end{array}
\right)+H_X\Bigg\rangle ,\end{equation}
which in the basis of Eq.~(\ref{eq:BaconShorBasis}) is
\begin{equation}\ker \Par_X=\Bigg\langle \left(
\begin{array}{c}
 1 \\
 1\\
 1\\
 1\\
\end{array}
\right)\Bigg\rangle .\end{equation}
Thus, for the Bacon--Shor code, $\Par_X$ is the parity check of the \(l\)-fold repetition code, which coincides with the $l$-fold repetition code described in the original work \cite{B06}. To correct an \(X\)-type error \(a\in G\), we first determine its syndrome \(\Syn_X(a)\) by measuring the $Z$-type gauge generators in Eq.~(\ref{eq:BaconShorCZ}), as described in Section \ref{sec:CSSRecoverySyndrome}. We then express \(\Syn_X(a) = a+\parenth{H_X+H_Z^{\theta}}\) in the basis of Eq.~(\ref{eq:BaconShorBasis2}). Next, by viewing $\Par_X$ as the matrix in Eq.~(\ref{eq:ParXMatrixBS}), we compute \(\left(a+H_X\right) + \ker \Par_X\) in the basis of Eq.~(\ref{eq:BaconShorBasis}). Now, suppose that the original \(X\)-type error \(a\) has odd weight in less than \(\frac{l}{2}\) rows. Then \(a+H_X\) has weight less than \(\frac{l}{2}\) in the basis of Eq.~(\ref{eq:BaconShorBasis}). Thus, since the distance of \(\Par_X\) is \(l\), we can determine \(a+H_X\) in the basis of Eq.~(\ref{eq:BaconShorBasis}). This allows us to correct the $X$-type error $a$ up to gauge terms.
\end{example}

\section{Parameterizing the codespace of a subsystem CSS code}\label{app:parameterize}

In this appendix, we parameterize the codespace of a subsystem CSS code $H = H_X \times H_Z$ in terms of the underlying classical codes $H_X$ and $H_Z$. The main result is as follows.

\begin{proposition}
Let $H = H_X \times H_Z$ be a subsystem CSS code. Write
\begin{align}
    {\parenth{H_X + H_Z^\theta}} / H_X &= {\set{l_i + H_X}}_i \text{ and} \\
    H_X / {\parenth{H_X \cap H_Z^\theta}} &= {\set{g_j + H_X \cap H_Z^\theta}}_j,
\end{align}
where all elements on the right-hand side are distinct. For any $a \in G$, write
\begin{equation}
    \ket{a + H_X \cap H_Z^\theta} = \sum_{s \in H_X \cap H_Z^\theta} \ket{a + {s}}.
\end{equation}
Then the codespace of $H$ is
\begin{equation}
    \dsL \otimes \dsG = \spn {\set{\ket{l_i + g_j + H_X \cap H_Z^\theta}}_{i,j}}.
\end{equation}
\end{proposition}
\begin{proof}
Say that $H$ has parameters $[[n,k,r,d]]$. From the standard theory of stabilizer codes, the codespace $\dsL \otimes \dsG$ has dimension $p^k p^r$ \cite{Scherer2019}. Moreover, one verifies that every state of the form
\begin{equation}
    \ket{l_i + g_j + H_X \cap H_Z^\theta}
\end{equation}
is fixed by every stabilizer of $H$ and so belongs to the codespace $\dsL \otimes \dsG$. Thus, we must show that the set
\begin{equation}
    \set{\ket{l_i + g_j + H_X \cap H_Z^\theta}}_{i,j}
\end{equation}
contains exactly $p^k p^r$ distinct elements. To see this, note that by Proposition \ref{prop:CSSStructure}, the index $i$ ranges over an index set of size $p^k$, and the index $j$ ranges over an index set of size $p^r$. Moreover, suppose that
\begin{equation}
    \ket{l_i + g_j + H_X \cap H_Z^\theta} = \ket{l_{i'} + g_{j'} + H_X \cap H_Z^\theta}.
\end{equation}
Then
\begin{equation}
    l_i + g_j + H_X \cap H_Z^\theta = l_{i'} + g_{j'} + H_X \cap H_Z^\theta,
\end{equation}
which implies
\begin{equation}
    l_i + H_X = l_{i'} + H_X.
\end{equation}
Since the $l_i + H_X $ are distinct by assumption, we must have $l_i = l_i'$, so $i = i'$. This implies that
\begin{equation}
    g_j + H_X \cap H_Z^\theta = g_{j'} + H_X \cap H_Z^\theta,
\end{equation}
so we conclude that $g_j = g_{j'}$, so $j = j'$.
\end{proof}

\bibliographystyle{plainnat}
\bibliography{main}

\end{document}